%% file: main.tex
\documentclass{article}
\usepackage[a4paper]{geometry}
\usepackage[utf8]{inputenc}
\usepackage{amsmath,amssymb,amsfonts}
\usepackage{natbib}
\usepackage{amsthm}
\usepackage{authblk}
\usepackage{graphicx}
\usepackage{float}
\newtheorem{theorem}{Theorem}

\usepackage{hyperref}
\usepackage{cleveref}
\usepackage{bbm}
\usepackage{multirow}
\usepackage[table,xcdraw]{xcolor}
\hypersetup{
    colorlinks=true,
    linkcolor=blue,
    citecolor = blue
}

\usepackage[ruled]{algorithm2e}% pour algo

\newcommand{\monmodele}{CoOP-LBM}%Corrected Observation Process LBM

\newcommand{\btheta}{\boldsymbol{\theta}}
\newcommand{\balpha}{\boldsymbol{\alpha}}
\newcommand{\bbeta}{\boldsymbol{\beta}}
\newcommand{\bpi}{\boldsymbol{\pi}}
\newcommand{\blambda}{\boldsymbol{\lambda}}
\newcommand{\bmu}{\boldsymbol{\mu}}
\newcommand{\bZ}{\mathbf{Z}}
\newcommand{\bx}{\boldsymbol{x}}

\title{Disentangling the structure of ecological bipartite networks from observation processes}
\author[1]{Emre Anakok}
\author[2]{Pierre Barbillon}
\author[3]{Colin Fontaine}
\author[4]{Elisa Thebault}
\affil[1,2]{Université Paris-Saclay, AgroParisTech, INRAE, UMR MIA Paris-Saclay, 91120, Palaiseau, France.}
\affil[3]{Centre d'Écologie et des Sciences de la Conservation, MNHN, CNRS, SU, 43 rue Buffon, 75005 Paris, France}
\affil[4]{Sorbonne Université, CNRS, IRD, INRAE, Université Paris Est Créteil, Université Paris Cité, Institute of Ecology and Environmental Sciences (iEES-Paris), 75005 Paris, France}
\affil[1]{emre.anakok@agroparistech.fr}

\date{}                     %% if you don't need date to appear
\begin{document}
\maketitle

\section*{Abstract}

The structure of a bipartite interaction network can be described by providing a clustering for
each of the two types of nodes. Such clusterings are outputted by fitting a Latent Block Model
(LBM) on an observed network that comes from a sampling of species interactions.
However, the sampling is limited and possibly uneven. This may jeopardize the fit of the LBM
and then the description of the structure of the network by detecting structures resulting from
the sampling and not from actual underlying ecological phenomena.
If the observed interaction
network consists of a weighted bipartite network where the number of observed interactions
between two species is available, the sampling efforts for all species can be estimated and used
to correct the LBM fit.  We propose to combine an observation model that accounts for sampling
and an LBM for describing the structure of underlying possible ecological interactions.  We
develop an original inference procedure for this model, the efficiency of which is demonstrated
in simulation studies.
Its relevance and its practical interest are attested on a large dataset of plant-pollinator networks, as we observe structural change on most of the networks depending on whether observation processes are accounted for or not.

\textbf{Keywords:} Latent Block Model,  Sampling Effect, Stochastic Expectation Maximization, Nestedness, Modularity

\section{Introduction}

%\paragraph{Networks in ecology}
\subsection{Networks in ecology}

Analysing the structure of ecological networks has proven very enlightening to understand the functioning and response to perturbation of ecological communities, leading to a strong growth in the number of publications in the field \citep{ings_review_2009}. 
Networks have been widely used to study food webs describing trophic interactions among species, before being used for other types of ecological interactions requiring the study of bipartite networks \citep{ings_review_2009}, including mutualistic relationships such as between plants and pollinators \citep{Lara-Romero-2019,kaiser-bunbury_robustness_2010}, plant seeds dispersed by ants \citep{bluthgen2004bottom}, or antagonistic relationships such as between hosts and parasites \citep{hadfield_tale_2014}.

Several metrics are used to study the structure of bipartite ecological networks as a whole, such as nestedness   \citep{Lara-Romero-2019,de_manincor_how_2020,terry_finding_2020,fortuna_nestedness_2010} or modularity \citep{de_manincor_how_2020,terry_finding_2020,fortuna_nestedness_2010,guimera_missing_2009}. 
To go beyond these metrics that target a specific structure, we use LBM \citep{govaert2010latent}, which is the bipartite adaption of the SBM \citep{hollandStochasticBlockmodelsFirst1983}.
This allows one to highlight network structure in a very flexible way, without predefining a structure beforehand, except the existence of blocks. LBM can reveal a modular or a nested structure if present, but is not restricted to these particular structure.  
This model is increasingly used to study the structure of ecological networks \citep{terry_finding_2020,sander_what_2015,leger_clustering_2015,Kefi,oconnor_unveiling_2020,durand-bessart_trait_2023}. Whenever an LBM is fitted on a bipartite network, it produces two clusterings, one for each set of species. These clusterings can be seen as groups of interacting species, each one having a specific pattern of interaction with the other blocks. In this model, the probability that two species are in interaction depends on which groups the species belong to. Such model has allowed to reveal well-defined functional groups (i.e. species sharing both the way they interact with others and their traits linked with mobility and habitat) in complex ecological networks with trophic and non-trophic interactions \citep{Kefi}. Multiple methods exist to estimate the parameters of an LBM \citep{brault_estimation_2014,kuhn_properties_2020}.
 
%\paragraph{Sampling issues}
\subsection{Sampling issues}
 In ecology, sampling a plant pollinator network is labor-intensive and many datasets only reveal a subset of the existing interactions \citep{chacoff_evaluating_2012,jordano_sampling_2016}. The network can be sampled with different methods, such as timed observations \citep{Lara-Romero-2019,kaiser-bunbury_robustness_2010}, transects \citep{de_manincor_how_2020,magrach_plantpollinator_nodate} or study of pollen \citep{de_manincor_how_2020}. These methods have an asymmetrical focus on either insects or plants, and therefore only allow an asymmetrical sampling effort control that could bias the observed network. \cite{de_manincor_how_2020} have shown that different types of sampling on the same network can lead to different structures. Differences in networks could also be caused by the duration of the study \citep{ings_review_2009}, or the observed abundance of the different species. Further, the log-normal distribution of species abundance in ecological networks \citep{sugihara_minimal_1980} makes the sampling of interactions uneven as the observer is more likely to observe interactions of abundant species, and thereby over estimate the degree of abundant species \citep{baldridge_extensive_2016,fort_abundance_2016}. Such effects of sampling processes on ecological network structure have recently be shown to be stronger than ecological effects such as latitudinal gradient or the impact of anthropogenic disturbances \citep{dore_relative_2021}.

Completeness of sampling for a species is defined as the proportion of its observed interactions over the total number of   its actual interactions.
Some propositions to evaluate this completeness rely on either  accumulation curves, which model the rate of new interaction observations over time  \citep{rivera2012effects},
or through external data assessing the abundance of the different interacting species \citep{bluthgen2004bottom}.
The sampling process can induce enormous biases in the statistical analyses of the networks that have not been taken into account in most analyses of ecological network structure. This raises doubts regarding the current understanding of the structure of pollination networks \citep{bluthgen2008interaction}.
In particular, recent studies debate whether the observed nestedness in interaction networks simply results from sampling effects \citep{staniczenko2013ghost,krishna_neutral-niche_2008}.
Moreover, the impact of sampling completeness on many metrics such as modularity or nestedness has been highlighted. \citep{rivera2012effects}. 
We demonstrate in the following motivating example how a credible  scenario of plant-pollinator interaction sampling effects leads to discovering an artificial structure in the data.

%\paragraph{Motivating example}
\subsection{Motivating example}This example is illustrated in \cref{fig0}. All the following elements are properly defined later in the paper.  Let us consider a binary bipartite interaction network that can be represented by its incidence matrix (black dots for interaction and white dots otherwise). A simple probabilistic model to generate such a network is based on the Erdős–Rényi model \citep{erdosrenyi1959}
where the probability of connection between two species is constant and the connections are drawn independently. The Erdős–Rényi model is considered as unstructured and is equivalent to the LBM with a unique block per set of species. 
If the interactions are sub-sampled uniformly, the inference of an LBM on this network data leads to finding a unique block with a lower estimated probability of connection. However, if the sub-sampling is not uniform, a structure may be detected.  In the simulation setting in \cref{fig0}, the intensity of the sampling depends on each species, for example due to differences in species abundances or unequal sampling effort among species. Even if most of the interactions have been sampled ($\sim 70\%$), a structure is found with two blocks when inferring an LBM (fig. 1C). This structure which is only caused by the sampling effort distribution could be misinterpreted as a true ecological phenomenon. Most studies using LBM do not take into account the potential missing interactions in the inference of the structure. 
%In this paper, we aim to propose a method that would correct such sampling bias and estimate the true underlying structure of the interaction matrix. 

% \begin{figure}%[H]
%     \begin{minipage}{.50\textwidth}
%     \centering
%     \includegraphics[width=1\textwidth]{image/ggplot_intro_0.4.png}
%     \end{minipage}
%     \begin{minipage}{.5\textwidth}
%     \centering
%     \includegraphics[width=1\textwidth]{image/ggplot_intro_1.png}
    
%     \end{minipage}
%     \begin{minipage}{.5\textwidth}
%     \centering
%     \includegraphics[width=1\textwidth]{image/ggplot_intro_2.png}
    
%     \end{minipage}
    
%     \caption{Example of a situation where an unstructured matrix which has been sub-sampled in a non-uniform way leads to a seemingly nested LBM. \textbf{A)} Full matrix of interactions. In this setup, all the interactions are possible and happen with a fixed probability. \textbf{B)} Sub-sampled matrix with 70\% interactions remaining, however the missing interactions are not spread uniformly, due to different sampling efforts. This kind of sub-sampling could appear because of the sampling protocol, abundance or the sampling time. \textbf{C)} Estimated clustering by LBM on the sub-sampled matrix, the identified groups being separated by red lines. 
%    }
%     \label{fig0}
%     \end{figure}

\begin{figure}
    \centering
    \includegraphics[width=0.85\textwidth]{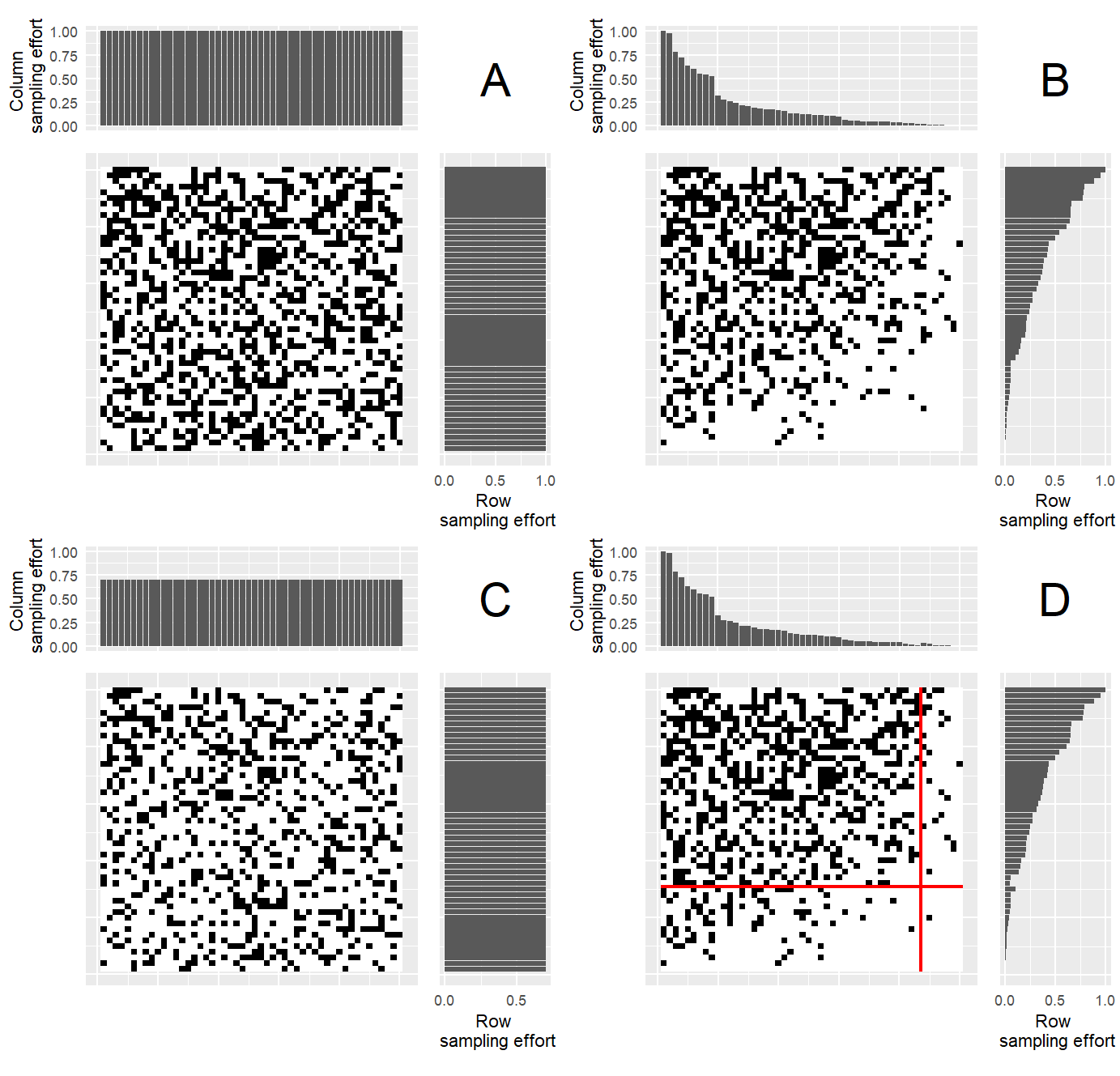}
    \caption{Example of a situation where an unstructured matrix which has been sub-sampled in a non-uniform way leads to a seemingly nested LBM. \textbf{A)} Full matrix of interactions. In this setup, all the interactions are possible and happen with a fixed probability. \textbf{B)} Sub-sampled matrix with 70\% interactions remaining, however the missing interactions are not spread uniformly, due to different sampling efforts. This kind of sub-sampling could appear because of the sampling protocol, abundance or the sampling time. \textbf{C)} Uniformly sub-sampled matrix with 70\% interactions remaining. This kind of sub-sampling is common in a simulation setting.  \textbf{D)} Estimated clustering by LBM on the non-uniformly sub-sampled matrix \textbf{B)}, the identified groups being separated by red lines. }
    \label{fig0}
\end{figure}

%\paragraph{Our contribution}
\subsection{Our contribution}
Both existing literature and the motivating example highlight the need for taking the observation process into account. 
We focus on cases where the data at hand only consists of a weighted network where the counts of interaction between pairs of species are reported.  No external data such as species abundance, reports of time of observation, traits nor phylogeny are available.
This kind of situation is rather common among the available datasets. 
Our underlying assumption is that the observed weighted network is a result of two phenomena: an ecological phenomenon that could be considered as a binary network providing the possible or impossible interactions and an observation phenomenon which is related to the relative abundances of species and their involvement in pollination interactions 
at the sampling time, and also related to the sampling protocol. 

The latter phenomenon results in smaller or larger counts and possibly to zero counts for actual but not observed interactions. 
Our goal is then to disentangle the two phenomena to recover in particular the structure of the binary network which accounts for the 
possible ecological interactions. %\ET{citer mckenzie ici ?}
%Following the approach used on other types of ecological data,}
In order to achieve this goal, we developed a new latent variable probabilistic model called the Corrected Observation Process for Latent Block Model (CoOP-LBM). This model assumes that the observed data is the product of a binary LBM which represents the interactions that are possible or not, with a sampling scheme following Poisson distributions. In this setting, non observed interactions (a $0$ in the incidence matrix) can have two sources: they can come from the binary network (impossible interaction) or they can come from the sampling scheme (missed interaction), but we cannot distinguish impossible interactions from missed interactions. 
We prove the identifiability of the model and we propose an algorithm to estimate the parameters of this model based on a Stochastic version of the Expectation-Maximisation (SEM) algorithm \citep{celeux1992stochastic}. Our SEM alternates between  maximisation steps, where parameters are updated, and two simulation steps, during which the location of missing interactions and the block belongings are simulated. An integrated complete likelihood criterion \citep{biernacki2000} is proposed to select the number of groups for the clustering. A simulation study is performed where we show the performance of our method to recover the clusters and to detect missing interactions.An R package is available at \url{https://github.com/AnakokEmre/CoOP-LBM}.  All the simulations are reproducible and available on \url{https://anakokemre.github.io/CoOP-LBM/index.html}
We also show how the imputation of missing interactions can correct the estimation of some metrics commonly used in ecology.
We applied our method on a large dataset \citep{dore_relative_2021}
that contains 70 networks with counting data. We show the goodness of fit of our model to the data and show that the detected structure is rather different when using a CoOP-LBM compared to an LBM on binarised networks.

%\paragraph{Related work} 
\subsection{Related work}
Several studies have started to include missing data within SBM, but often in cases where the missing interactions are properly labelled.

Estimation of the unipartite stochastic block model (SBM) with missing data has been studied for the binary case by \cite{tabouy_variational_2019}, adapting a variational expectation maximization algorithm, and for the weighted case by \cite{Aicher_2014} using a Bayesian approach. Contrary to our study, both have access to where the data is missing.
A few recent studies in ecology also aimed to estimate missing interactions. \cite{terry_finding_2020} provide and compare different methods of estimation and imputation of missing interactions. To estimate the number of missing interactions, \cite{terry_finding_2020} and \cite{macgregor_estimating_2017}  use the Chao estimator \citep{chao_coverage-based_2012}, which will be compared to our method. %\cite{mackenzieEstimatingSiteOccupancy2002} propose a method for estimating site occupancy rates even if the species is undetected, but does not consider the latent block structure of the LBM.}
More recently, \cite{papadogeorgou_covariate-informed_2023} have developed a model
which is fitted from the observation of several networks of interactions between birds and seeds where the involved species are redundant and where additional information on species such as traits and phylogeny is available.
Their model is focused on the prediction of potential interactions and not on unraveling a block structure contrary to our model which only requires the observation of a single weighted network of interactions with no additional information.
Our model differs from the degree-corrected SBM (DCSBM) \citep{karrerStochasticBlockmodelsCommunity2011}  and its bipartite counterpart \citep{zhao_variational_2022}, which build their clusters according a Poisson likelihood, but do not take into account the possibility of missing an interaction.
\cite{tabouy_variational_2019} also demonstrated how the inference of the  Stochastic Block Model is jeopardized by the observation process.

%\paragraph{Outline of the paper}
\subsection{Outline of the paper}
We give the mathematical definition of the model and prove its identifiability in \cref{Definition}. Then, \cref{inference} details the inference algorithm and the model selection procedure. The simulation study is presented in \cref{simulation}. All the simulations are reproducible and available on \url{https://anakokemre.github.io/CoOP-LBM/index.html}. 
Eventually, the application to 70 plant-pollinator networks is done and discussed in \cref{secappdata}.

\section{Definition of the model}\label{Definition}

The available data are given as a matrix  $R = (R_{i,j})_{\substack{ i = 1, \dots, n_1\\ j = 1, \dots, n_2}}$ where the $n_1$
rows correspond to a type of species (plants e.g.) and the $n_2$ columns correspond to the other type of species in interaction (pollinators e.g.). For all $i = 1, \dots, n_1$ and $j = 1, \dots, n_2$, the elements $R_{ij}\in\mathbb{N}$ since they denote the 
counts of observed interactions between species $i$ and $j$.

We assume that $R = M \odot N $ is the Hadamard product of two matrices, 
with $M$ being a realization of an LBM and  the elements of $N$ follow  Poisson distributions independent on $M$. More precisely, for $M$,
we assume that latent variables  $\bZ^1 = (Z^1_{i})_{ i = 1,\dots, n_1} \in \{1,\dots,Q_1\}^{n_1}$ and $\bZ^2 = (Z^2_{j})_{ j = 1,\dots, n_2} \in \{1,\dots,Q_2\}^{n_2}$ 
provide respectively a partition of species in rows and in columns, each species in row or in column only belonging to one cluster.
%The variable $Z^1_{1k}$ is a binary indicator such that $Z^1_{ik}=1$ if species $i$ (in row) belongs to block $k$, $0$ otherwise. The same for $Z^2_{jl}=1$ if species $j$ (in column) belongs to cluster $l$, $0$ otherwise. 
%We assume that these latent variables are independent and that
%Let $A:= \{\balpha \in \mathbb{R}^{Q_1} | \sum_{k=1}^{Q_1}\alpha_k = 1 , \alpha_k>0\}$ be the $Q_1$ dimensional simplex, and let $B$ be the $Q_2$ dimensional simplex.} %We assume that 
%$\mathbb{P}(Z^1_{i}=k)=\alpha_k$ for all $1\le i\le n_1$, $1\le k\le Q_1$, with $\balpha \in A$ and $\mathbb{P}(Z^2_{j}=l)=\beta_l$ for all $1\le j\le n_2$, $1\le l\le Q_2$, with $\bbeta \in B$. 
We assume that $Z^1_1,\dots,Z^1_{n_1}$ are independent and identically distributed random variables following a multinomial distribution $Mult(1,\balpha)$, with $\balpha \in A=\{\balpha=(\alpha_1,\ldots,\alpha_{Q_1}) \in \mathbb{R}^{Q_1} | \sum_{k=1}^{Q_1}\alpha_k = 1 , \text{ and for all }k,\ \alpha_k>0\}$, such that $\mathbb{P}(Z^1_{i}=k)=\alpha_k,1\le k\le Q_1$. Similarly, we assume that $Z^2_1,\dots,Z^2_{n_2}$ are independent and identically distributed random variables following a multinomial distribution $Mult(1,\bbeta)$, with $\bbeta \in B=\{\bbeta=(\beta_1,\ldots,\beta_{Q_2}) \in \mathbb{R}^{Q_2} | \sum_{l=1}^{Q_2}\beta_l = 1 , \text{ and for all }l,\ \beta_l>0\}$ such that $\mathbb{P}(Z^2_{j}=l)=\beta_l, 1\le l\le Q_2$. Moreover, the latent variables in $\bZ^1$ are also independent of the latent variables in  $\bZ^2$.

Given the latent variables, the distribution of the elements of $M$ is given by 
\begin{equation} 
\mathbb{P}(M_{i,j} = m | Z^1_{i}=k, Z^2_{j}=l) = \pi_{kl}^m (1-\pi_{kl})^{1-m}, \quad m\in\{0,1\}. 
\end{equation}

Moreover, for any integer $v$ let $\Pi_{v} = \{\bx=(x_1,\ldots,x_v) \in (0,+\infty)^{v} | \max_{1\le i\le v} x_i=1 \}$.
We assume that the random variables $(N_{i,j})_{1\le i\le n_1, 1\le j\le n_2}$ are independent, independent of the random variables $(M_{i,j})_{1\le i\le n_1, 1\le j\le n_2}$ and such that, for all $1\le i\le n_1,\ 1\le j\le n_2$, $N_{i,j} \sim \mathcal{P}oisson(\lambda_i\mu_jG) := \mathcal{P}(\lambda_i\mu_jG)$ with $\blambda \in\Pi_{n_1} $, ${\bmu\in \Pi_{n_2}}$ and $G>0$. The parameter $\lambda_i$ can represent the relative sampling effort, the apparent abundances or the relative detectability of row species $i$, $\mu_j$ represents the same for column species $j$, and $G$ is a constant representing the %global sampling effort of the network. 
maximum sampling effort across all pairs of species.

Therefore, for all $1\le i\le n_1,\ 1\le j\le n_2$, the probability distribution of $R_{i,j} = M_{i,j} N_{i,j}$ given the latent variables $\bZ^1,\bZ^2$ is
\begin{equation}
\mathbb{P}(R_{i,j} = r| Z^1_{i}=k, Z^2_{j}=l;\btheta) = \left\{
    \begin{array}{ll}
        \ \pi_{kl} \frac{(\lambda_i\mu_jG)^{r}}{r!}e^{-\lambda_i\mu_jG}& \mbox{if } r>0 \\
        \\
        1 - \pi_{kl} ( 1 - e^{-\lambda_i\mu_jG})& \mbox{if }  r=0
    \end{array}
\right.
\end{equation}

 where $\btheta = (\balpha,\bbeta,\bpi,\blambda,\bmu,G)$ with 
 $\balpha=(\alpha_k)_{1\le k\le Q_1}\in A$, $\bbeta=(\beta_l)_{1\le l\le Q_2}\in B$, $\bpi=(\pi_{kl})_{1\le k\le Q_1,1\le l\le Q_2}\in[0,1]^{Q_1\times Q_2}$, $\blambda=(\lambda_i)_{1\le i\le n_1}\in\Pi_{n_1}$, ${\bmu=(\mu_j)_{1\le j\le n_2}\in\Pi_{n_2}}$, $G>0$, 
 denotes all the parameters to be estimated.
The full log-likelihood is then 

$$\log\mathcal{L}(\btheta;R,\bZ^1, \bZ^2) = \log\mathcal{L}(\balpha;\bZ^1) + \log\mathcal{L}( \bbeta;\bZ^2) + \log\mathcal{L}(\bpi,\blambda,\bmu,G;R|\bZ^1, \bZ^2) $$ 
with

\begin{equation}
    \log\mathcal{L}(\balpha,\bZ^1)= \sum_{i=1}^{n_1}\sum_{k=1}^{Q_1} \mathbbm{1}_{\{Z^1_i = k\}} \log(\alpha_k), \quad
    \log\mathcal{L}(\bbeta;\bZ^2) =\sum_{j=1}^{n_2}\sum_{l=1}^{Q_2}\mathbbm{1}_{\{Z^2_j = l\}}  \log(\beta_l)
\end{equation}

\begin{align}
&\log\mathcal{L}(\bpi,\blambda,\bmu,G;R|\bZ^1, \bZ^2)  = \\ &\sum_{\substack{i,j\\R_{i,j}>0}}
\sum_{k=1}^{Q_1}\sum_{l=1}^{Q_2}
\mathbbm{1}_{\{Z^1_i = k\}} \mathbbm{1}_{\{Z^2_j = l\}} \biggl(log(\pi_{kl}) + R_{i,j}\log(\lambda_i\mu_j G) - \lambda_i\mu_j G-\log(R_{i,j}!)\biggr)\notag \\
    + &\sum_{\substack{i,j\\R_{i,j}=0}}\sum_{k=1}^{Q_1}\sum_{l=1}^{Q_2} \mathbbm{1}_{\{Z^1_i = k\}}\mathbbm{1}_{\{Z^2_j = l\}}\biggl(\log(1-\pi_{kl}(1-e^{\lambda_i\mu_j G}))\biggr)\notag\,.
\end{align}

Since the variables $\bZ^1$ and $\bZ^2$ are latent, they are integrated out in order to obtain the observed log-likelihood:

\begin{equation}\label{eqobslik}
 \log \mathcal{L}(\btheta;R) = \sum_{(\bZ^1,\bZ^2) \in \{1,\ldots,Q_1\}^{n_1}\times \{1,\ldots,Q_2\}^{n_2}}\log\mathcal{L}(\btheta;R,\bZ^1, \bZ^2)\,.
\end{equation}

This model is identifiable according to the following theorem.
\begin{theorem}
Sufficient condition for identifiability. 
Under the following assumptions on the parameters in $\btheta$ and the size of the network:
\begin{itemize}
    \item (A1) for all $1 \leq k \leq Q_1,\alpha_k  > 0$ and the coordinates of vector $\tau = \bpi \bbeta$ are distinct,
     \item (A2) for all $1 \leq k \leq Q_2,\beta_k  > 0$ and the coordinates of vector $\sigma= \balpha^\top \bpi $ are distinct (where $\balpha^\top$ is the transpose of $\balpha$),
     \item (A3) $n_1 \geq 2Q_2 -1$ and $n_2 \geq 2Q_1-1$ ,
     \item (A4) $G>0$, $\blambda\in \Pi_{n_1}$, $\bmu \in \Pi_{n_2}$
\end{itemize}

then \monmodele{} is identifiable. 
\end{theorem}

\begin{proof}{}
  The proof is given in the appendix and is in two parts: the first part deals with the identifiability of $(\blambda,\bmu,G)$, the second part, adapted from the proof of \cite{celisse_consistency_2012} and \cite{brault_estimation_2014}, deals with the identifiability of $(\balpha,\bbeta,\bpi)$.
\end{proof}

\paragraph{Ecological justification of the CoOP-LBM}
%\subsection{Ecological justification of the CoOP-LBM}
The CoOP-LBM connects well with the way species interaction networks are viewed in ecology because ecological studies make a clear distinction between the "true" network describing all possible interactions among species and the observed network resulting from sampling process and the relative abundance of the species at a given location.

On other types of ecological data such as the distribution of species in space, similar approaches combining a model for the ecological phenomenon and a model for the observation process is widely used this approach  that combines a model for the ecological phenomenon and a model for the detection during the sampling process is widely used \citep{mackenzieEstimatingSiteOccupancy2002,doser2022a,doser2022b}.
In the context of pollination networks, unobserved interactions have been categorized as either forbidden interactions (i.e. the interaction cannot occur because species do not co-occur or have mismatching traits) or missing interactions (i.e. the interaction exists but would require additional sampling to be detected) \citep{olesen_missing_2011,jordano_sampling_2016}. Our model of interaction sampling process also relates well with classical assumptions and results on species interactions in ecology. Abundances of plants, and pollinators, which could relate in our model with species relative sampling effort denoted by $\lambda_i$ and $\mu_j$, are known to greatly correlate with interaction frequency in pollination networks \citep{fort_abundance_2016}. 
The probability of interaction among species is moreover often assumed to be directly proportional to the product of species relative abundances in ecological networks under the mass action hypothesis  \citep[e.g.][]{staniczenko2013ghost}, hence this corresponds to the expected value of the Poisson distribution in our model: $\lambda_i \mu_j G$, with $G$ being the maximum sampling effort across all pair of species. It should also be noted that our model assumes no species preference, which could also affect species interaction probability in ecological networks \citep{ staniczenko2013ghost}. Defining a preference parameter by species pair $\lambda_{i,j}$ would not lead to a unique estimation for observations where $R_{i,j}=0$, and could over-parameterize the model. This model would be interesting if we had access to $\lambda_{i,j}$ through a proxy. Such a proxy could rely on using a set of covariates on plants and insects describing different species traits, such as mouthpart and corola length for insect and plant, respectively, that would allow the estimation of a similarity measure or a distance measure among the interacting partners.
However such traits related to the choice of interacting partners are not available for most datasets of plant pollinator networks, which only include matrices, matrices providing the counts of observed interactions, e.g. \cite{dore_relative_2021} the traits governing the species preferences remain to be identified for plant-pollinator networks.
As such, the majority of datasets of plant-pollinator networks only include matrices, providing the counts of observed interactions, without covariates, e.g. \cite{dore_relative_2021}. Additionally,  we have studied the robustness of our model to different misspecifications with simulations, to confirm that the model is still able to perform correctly in various settings.

% Combining
% On other types of ecological data such as the distribution of species in space, \cite{mackenzieEstimatingSiteOccupancy2002} propose a method for estimating site occupancy rates that accounts for the observation process by 
% This type of approach is used on other types of ecological data with models such as \cite{mackenzieEstimatingSiteOccupancy2002}'s occupancy model with imperfect detection, where a method is proposed for estimating site occupancy rates even if the species is undetected, but the model does not consider the latent block structure of the network.}

\section{Inference and model selection}\label{inference}

From the observation of an interaction network where the counts of occurring interactions are recorded, our goal is to infer our CoOP-LBM.
This inference can be separated into two parts: inferring the parameters when the numbers of blocks in rows ($Q_1$) and columns ($Q_2$) are known, selecting these numbers of blocks. The former part cannot be dealt with by directly maximising the observed likelihood given in Equation \eqref{eqobslik} since the sums over the set $\{1,\ldots,Q_1\}^{n_1}\times\{1,\ldots,Q_2\}^{n_2}$ quickly becomes intractable as $Q_1$, $Q_2$, $n_1$ or $n_2$ grow.
A classical solution in mixture models with latent variables is to make recourse to an Expectation-Maximisation (EM) algorithm \citep{dempster1977maximum}. However, due to dependencies between the latent variables when conditioned to the observed data, the EM algorithm is not practicable  in the SBM or the LBM.
Several alternatives have been proposed such as a variational approximation of the EM (VEM) algorithm \citep{daudin_mixture_2008,govaert2008block} or stochastic version of the EM algorithm \citep{brault_estimation_2014,kuhn_properties_2020}.
In Section \ref{ssecInf}, 
for a given $Q_1$ and $Q_2$, we provide our inference procedure which is inspired by a version of a stochastic EM algorithm \citep{brault_estimation_2014}. 
In Section \ref{ssecMod}, we derive a penalised likelihood criterion to 
select the numbers of blocks $Q_1$ and $Q_2$. This criterion is an Integrated Classification Likelihood (ICL) which has proven its practical efficiency in many blockmodels \citep{daudin_mixture_2008,brault_estimation_2014}.

\subsection{Estimation of parameters}\label{ssecInf}
The estimation of the parameters given $Q_1$ and $Q_2$ follows this scheme described in Algorithm \ref{algo:vem:nmar}. The steps are detailed in the text hereafter.

\begin{algorithm}[h]
  \SetSideCommentLeft
  \DontPrintSemicolon
  \KwSty{Initialisation:} Provide $\bZ^1_{(0)},\bZ^2_{(0)}, \bpi_{(0)},\Tilde{M}_{(0)}$.\;
  \Repeat{Convergence or max number of iterations reached.}{
    \begin{enumerate}
        
        \item M-step a) :  update  $\balpha_{(n+1)},\bbeta_{(n+1)}|\bZ^1_{(n)}, \bZ^2_{(n)}$, 
        \item
    M-step b) : update  $\blambda_{(n+1)},\bmu_{(n+1)},G_{(n+1)}|\Tilde{M}_{(n)}$, 
    \item
    S-step a): simulate $ \Tilde{M}_{(n+1)} |\bZ^1_{(n)}, \bZ^2_{(n)} ,\bpi_{(n)},\blambda_{(n+1)},\bmu_{(n+1)},G_{(n+1)}$, 
    \item
    M-step c) : update $\bpi _{(n+1)} |\Tilde{M}_{(n+1)},\bZ^1_{(n)}, \bZ^2_{(n)}$ , 
    \item
     S-step b) :  simulate  $\bZ^1_{(n+1)}, \bZ^2_{(n+1)} |\balpha_{(n+1)},\bbeta_{(n+1)},\bpi_{(n+1)},\Tilde{M}_{(n+1)}$.
    \end{enumerate}
  }
 \caption{Stochastic EM for \monmodele{} inference}
 \label{algo:vem:nmar}
\end{algorithm}

The initial clusterings of nodes $\bZ^1_{(0)}$ and $\bZ^2_{(0)}$ can be obtained with hierarchical, spectral or k-means clustering algorithms. The initial matrix of probabilities of connection $\bpi_{(0)}$ is computed as the maximum likelihood estimator with $\bZ^1_{(0)}$ and $\bZ^2_{(0)}$ given, and $\Tilde{M}_{(0)}$ is initialized with  the matrix $(\mathbbm{1}_{\{R_{i,j}>0\}})_{1\le i\le n_1,1\le j\le n_2}$. 

After the burn-in period, the algorithm ends when $||\hat{\btheta}_{(n)}-\hat{\btheta}_{(n+1)}||_2< \epsilon $ with $\hat{\btheta}_{(n)} = \frac{1}{n} \sum_{i=1}^n {\btheta}_{i}$ or when the maximum number of iterations is reached.
Convergence of similar stochastic EM  algorithms to a local maximum has been proved by \cite{delyon_convergence_1999}.

\paragraph{M-Step a)}
To update $\balpha,\bbeta$ given $\bZ^1_{(n)}, \bZ^2_{(n)}$, we use the
 maximum likelihood estimators:
 
\begin{equation}\forall 1\le k\le Q_1,\ 1\le l\le Q_2,\quad \alpha_k= \frac{1}{n_1}\sum_{i=1}^{n_1} \mathbbm{1}_{\{Z^1_i = k\}},
\quad \beta_l =\frac{1}{n_2} \sum_{j=1}^{n_2}
\mathbbm{1}_{\{Z^2_j = l\}}. \end{equation} 
\paragraph{M-Step b)}

$\lambda$ is updated with the following fixed point algorithm : 
\begin{equation}
\lambda \propto \frac{\sum_{j=1}^{n_2}R_{k,j}}{\sum_{j=1}^{n_2}m_{k,j}\frac{\sum_{i=1}^{n_1}R_{i,j}}{\sum_{i=1}^{n_1}m_{i,j}\lambda_{j}}}.\end{equation}

As $max_i \lambda_i =1$, a solution $\lambda$ that has been found by the fixed point algorithm can be divided by its maximum coordinate in order to have an estimation of $\lambda$. Once $\lambda$ calculated, we estimate \begin{equation}
G\mu_l= \frac{\sum_{i=1}^{n_1}R_{i,l}}{\sum_{i=1}^{n_1}m_{i,l}\lambda_{l}}.\end{equation}
 Finally, we deduce $G$ and $\mu$ by fixing $\max_j \mu_j = 1$.
As the matrix $M$ is not observed, it is replaced by a simulated version of it  $\Tilde{M}_{(n)}$.
Additional details are given in the supplementary material \ref{M-Steb b}.

\paragraph{S-Step a)}
The purpose of $\Tilde{M}_{(n)}$ is to distinguish zeros that come from the $M$ matrix (impossible interaction) and the $N$ matrix (missing interaction) by completing all zeros of the matrix $(\mathbbm{1}_{\{R_{i,j}>0\}})_{i,j}$ by a Bernoulli variable with probability $\mathbb{P}(M_{i,j}=1|R_{i,j}=0,\lambda_i, \mu_j,G,\pi, Z^1,Z^2)$. 

\begin{align}
&\mathbb{P} (M_{i,j}=1|R_{i,j}=0,\lambda_i, \mu_j,G,\pi, Z^1_{i}=k,Z^2_{j}=l)\notag\\
& \notag\\
 = &\frac{\mathbb{P}(M_{i,j}=1 , R_{i,j}=0 |\lambda_i, \mu_j,G,\pi, Z^1_{i}=k,Z^2_{j}=l)}{\mathbb{P}(R_{i,j}=0 |\lambda_i, \mu_j,G,\pi,Z^1_{i}=k,Z^2_{j}=l)}\notag\\
=&\frac{\pi_{kl}e^{-\lambda_i \mu_j G}}{1-\pi_{kl}(1-e^{-\lambda_i\mu_jG})}.
\end{align}

$\Tilde{M}_{(n+1)}$ is simulated with the following scheme:
$
\Tilde{M}_{i,j} = 
    \begin{cases}
    1 ,& \text{if } r_{i,j}> 0\\
    Y_{i,j},& \text{if } r_{i,j}= 0
    \end{cases}
$
where $Y_{i,j}$ follows a Bernoulli distribution with parameter
$\frac{\pi_{kl}e^{-\lambda_i  \mu_j G}}
{1-\pi_{kl}(1-e^{-\lambda_i\mu_j G})}\,.$
\paragraph{M-Step c)}
Update $\bpi_{(n+1)}$ with the maximum likelihood estimator given $M$, where $M$ is replaced by $\Tilde{M}$

\begin{equation}
\forall 1\le k\le Q_1,\ 1\le l\le Q_2,\quad   \pi_{kl}= \frac{\sum_{i,j}\mathbbm{1}_{\{Z^1_i = k\}}\mathbbm{1}_{\{Z^2_j = l\}}\Tilde{M}_{i,j}}{\sum_{i,j}\mathbbm{1}_{\{Z^1_i = k\}}\mathbbm{1}_{\{Z^2_j = l\}}}\,. \end{equation}

\paragraph{S-Step b)}
The latent variables $\bZ^1_{(n+1)}$ given $\bZ^2_{(n)}$ and the parameters are simulated first then the other latent variables $\bZ^2_{(n+1)}$ given $Z^1_{(n+1)}$ and the parameters are simulated.
Then for all $1\le k \le Q_1$, $1\le i\le n_1$:

% \begin{align}
%    &\mathbb{P}(Z_{i}^1=k}|R,\btheta,\bZ^2)\notag\\ &\propto \mathbb{P}(R |\btheta,Z^1_{i}=k},\bZ^2)\mathbb{P}(Z^1_{i}=k}) \notag\\
%    &\propto \alpha_k\prod_{j=1}^{n_2}\prod_{l=1}^{Q_2} \left(\pi_{kl} \left(\frac{e^{-\lambda_i\mu_jG}}{R_{i,j}}(\lambda_i\mu_jG)^{R_{i,j}} \right)^{\mathbbm{1}_{R_{i,j}>0}}\left(1-\pi_{kl}(1-e^{-\lambda_i\mu_jG}) \right)^{\mathbbm{1}_{R_{i,j}=0}}\right)^{\mathbbm{1}_{Z_{jl}=1}}.
% \end{align}
\begin{align}
   &\mathbb{P}(Z_{i}^1=k|R,\btheta,\bZ^2) \propto \mathbb{P}(R |\btheta,Z^1_{i}=k,\bZ^2)\mathbb{P}(Z^1_{i}=k) \\
   &\propto \alpha_k\prod_{j=1}^{n_2}\prod_{l=1}^{Q_2} \left(\pi_{kl} \left(\frac{e^{-\lambda_i\mu_jG}}{R_{i,j}}(\lambda_i\mu_jG)^{R_{i,j}} \right)^{\mathbbm{1}_{R_{i,j}>0}}\left(1-\pi_{kl}(1-e^{-\lambda_i\mu_jG}) \right)^{\mathbbm{1}_{R_{i,j}=0}}\right)^{\mathbbm{1}_{Z_{jl}=1}}.\notag
\end{align}
The simulation of $\bZ^2_{(n+1)}$ is done in a symmetrical way with $\bZ^1_{(n+1)}$ being given.

\paragraph{Algorithm output}
At the end, all the iterations following a ``burn-in'' phase  $\hat{\btheta}_{(n)}=(\balpha_{(n)},\bbeta_{(n)},\bpi_{(n)},\blambda_{(n)},\bmu_{(n)},G_{(n)})$   are averaged as $\hat{\btheta} = \frac{1}{n} \sum_n \hat{\btheta}_{(n)}$. Even if the methods allow us to give the probabilities of  belonging to each block,  a final hard clustering is determined with the majority rule given by $\underset{k}{argmax}\:  \mathbb{P}(Z^1_{i}=k)$.

\subsection{Model selection}\label{ssecMod}
In most cases, the number of groups $Q_1$ and $Q_2$ are not known.
 Then, in order to choose the number of blocks in a \monmodele{}, we propose to use an ICL criterion similarly to the ICL criteria developed in \cite{daudin_mixture_2008} or \cite{govaert2008block}. The parameter $\btheta = (\balpha,\bbeta,\bpi,\blambda,\bmu,G)$ lies in  $\Theta = A \times B \times \Pi_{n_1} \times \Pi_{n_2}  \times \mathbb{R}_{+}^\star$. Let $g(\btheta |m_{Q_1,Q_2})$ be the prior distribution of the parameters for a model $m_{Q_1,Q_2}$ with $Q_1$ and $Q_2$ classes. The ICL
criterion is an approximation of the complete-data integrated
likelihood defined as
\begin{equation}
\mathcal{L}(R,\bZ^1,\bZ^2|m_{Q_1,Q_2}) = \int_{\Theta} \mathcal{L}(\btheta;R,\bZ^1,\bZ^2|m_{Q_1,Q_2})g(\btheta|m_{Q_1,Q_2})d\btheta\,.\end{equation}

 \begin{theorem}
 For a model $m_{Q_1,Q_2}$ with $Q_1$ and $Q_2$ blocks the corresponding ICL criterion is:
 \begin{align}
     ICL(m_{Q_1,Q_2}) =& \notag \max_{\btheta}\mathcal{L}(\btheta;R,\widehat{\bZ^1}, \widehat{\bZ^2} | m_{Q_1,Q_2})\\
     &-\frac{Q_1-1}{2}\log(n_1) - \frac{Q_2-1}{2}\log(n_2) - \frac{Q_1Q_2 + n_1 +n_2-1}{2}\log(n_1n_2),
 \end{align}
 with $\widehat{\bZ^1}, \widehat{\bZ^2}$ the hard clusterings outputted by the inference algorithm.
 \end{theorem}

\begin{proof}
The proof is similar to the proof in \cite{daudin_mixture_2008}, and is given in the supplementary material \ref{ICL}. 
\end{proof}
 
The algorithm starts with a number of group of $Q_1=Q_2=1$ and fits independently the CoOP-LBM for increasing value of $Q_1$ and $Q_2$, each time initializing with a clustering algorithm such as spectral, hierarchical or k-means. The ICL is calculated for each model. If the ICL keeps decreasing as $Q_1$ or $Q_2$ increases, the algorithm can either stop, or deepens the exploration with a split-merge procedure. In the second case, the algorithm will initialize again by either merging different combinations of the previously estimated clusterings, or by splitting them. At the end, the algorithm returns a list containing the estimated model for each value of $Q_1$ and $Q_2$, with the corresponding ICL, and provide the user with the pair $(Q_1,Q_2)$ which has the larger ICL value.

\section{Simulations}\label{simulation}
\subsection{Settings}

The following simulation study is fully reproducible and can be performed with any set of chosen parameters. The code is available on \url{https://anakokemre.github.io/CoOP-LBM/index.html}.
In this simulation study, we consider that $n_1 = n_2 = 100, Q_1 = Q_2 = 3, \balpha = \bbeta = (\frac{1}{3},\frac{1}{3},\frac{1}{3})$ and 

$$\bpi= \begin{bmatrix}
0.95 & 0.75 & 0.50 \\
0.75 & 0.50 & 0.50 \\
0.50 & 0.50 & 0.05 
\end{bmatrix}  .$$ Moreover, we assumed that $G$ takes values in $\{25,100,200,\dots,600\}$ and that $\blambda =\blambda^*/\max(\blambda^*)$ with $ \blambda^* \overset{iid}{\sim} Beta(0.3,1.5)$, $\bmu=\bmu^*/\max(\bmu^*)$ with $ \bmu^* \overset{iid}{\sim} Beta(0.3,1.5)$. 

%The matrix $M$ follows an LBM with parameters $\balpha,\bbeta,\bpi$. The elements of the matrix $N_{i,j}$ are simulated following Poisson distributions with parameters $\blambda_i\bmu_j G$. The observed network is the term by term product $R_{i,j} = M_{i,j}\odot N_{i,j}$.
The matrices $M,\ N$ and $R$ are simulated according to the \monmodele{} with the parameters given above.
The rows and columns of $R$ and $M$ where there is no observation (full rows/columns of $R$ is filled with 0) are discarded.  The observed support is $V_{i,j} = \mathbbm{1}\{R_{i,j}>0\}$.

For each value of $G$, $10$ simulations of $R$ are performed. On each simulation, an LBM is fitted on the corresponding binary matrix $V$ using a VEM algorithm implemented and our proposed algorithm is fitted on $R$. 
$\widehat{Z^1}$ and $\widehat{Z^2}$ are initialized with a hierarchical clustering. Then, the algorithm runs for an initial burn-in of 50 iterations and then an additional 50 iterations is used to provide the estimator.

%\vspace{0.5cm}

The $Beta(0.3,1.5)$ distribution is skewed, there will be many values of $\lambda_i$ and $\mu_j$ close to $0$ and few close to $1$. In practice, this is what is observed in pollination networks, some species are very well sampled ($\lambda_i$ close to 1) and many of them are much less sampled ($\lambda_i$ close to 0). Consequently, the observed matrix will have a lot of missing data. For $G = 25$, at least $2/3$ of the support is missing in comparison with the full support, and for $G= 600$, almost half the support is missing.

\subsection{Results}

We verify that our algorithm is able to correctly estimate the parameters, and evaluate the performance of our model to recover the clusters, to detect missing interactions and to retrieve the values of several metrics computed with respect to the actual $M$ which is not observed in real situations  for varying sampling intensity $G$. \\

We remind that the bipartite degree-corrected SBM (DCSBM) \citep{karrerStochasticBlockmodelsCommunity2011,zhao_variational_2022} also provide us a clustering $Z^1,Z^2$ and estimated parameters $\Omega_{i,j} := v^1_i v^2_j \omega_{Z^1_i Z^2_j} $ such as 
\begin{equation}
    \mathbb{P}(R_{i,j}=r|Z_i^1=k,Z_j^2=l )= e^{-v^1_i v^2_j \omega_{k l} }\frac{(v^1_i v^2_j \omega_{k l} )^r}{r!} = e^{-\Omega_{i,j} }\frac{(\Omega_{i,j} )^r}{r!}
\end{equation}
where $v^1,v^2$ and $\omega$ have to satisfy certain constraints to be identifiable, and $\Omega_{i,j}$ is the expected value of the biadjacency matrix element $R_{i,j}$.
%\PB{que veut dire known matrix M c'est connu pcq on simule non ???, ou c'est retrieve the values of several metrics computed with respect to the actual M which is not observec in real situations ???}

\subsubsection{Estimation of $\blambda,\bmu, G$}

Our simulations show that \monmodele{} is able to correctly estimate $\blambda$,$\mu$ and $G$. For different value of $G$, a RMSE has been calculated for $\blambda$ and $\bmu$ with the following formulas:

\begin{equation}
RMSE_{\blambda} = \sqrt{ \frac{1}{n_1}\sum_{i=1}^{n_1}(\lambda_i - \hat{\lambda}_i)^2},  \quad RMSE_{\bmu} =\sqrt{ \frac{1}{n_2}\sum_{j=1}^{n_2}(\mu_j - \hat{\mu}_j)^2 } .\end{equation}

In supplementary material, we show that the parameters are well retrieved by our estimation procedure.

\subsubsection{Recovery of clusterings}
The clustering on rows and columns estimated by the CoOP-LBM are assessed by computing the Adjusted Rand Index  (ARI) \citep{rand} scores to the original blocks that were used to simulate $M$, and compared with the results estimated by the LBM and the DCSBM. When the ARI score is equal to one, the two clustering are exactly the same (up to label switching). When the ARI score is close to 0, then the two clustering are totally different. An ARI score can be computed even if the two clusterings do not have the same number of groups, the score is then necessarily lower than $1$.

\begin{figure}[H]
\centering
\includegraphics[width=10cm]{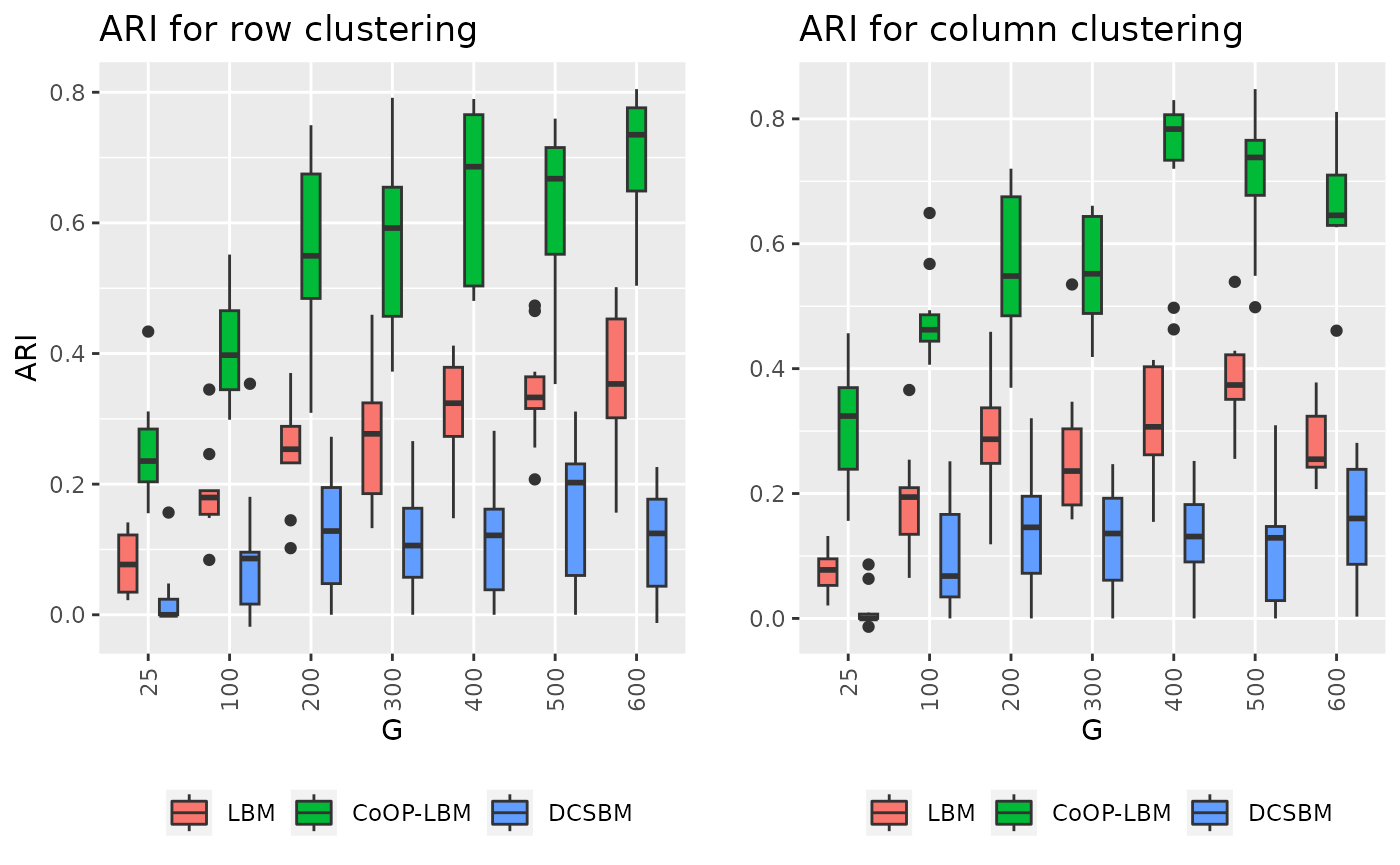}
\caption{ARI scores for rows and columns blocks when the number of groups is unknown}
\label{fig5}
\end{figure}

As we can see in \cref{fig5}, the \monmodele{} has systematically   a better ARI score in this setting than the LBM and the DCSBM, which tend to overestimate the number of groups compared to the CoOP-LBM as we explained in the introduction. The DCSBM also clusters the nodes according to the expected value of the Poisson parameters, which is not how the latent blocks are structured.
 Both LBM and DCSBM struggle to reach an ARI of $0.5$ for all values of $G$, whereas the \monmodele{} has an ARI score larger than $0.5$ when $G$ is larger than $250$. The LBM performs even worse when the number of groups is given (see supplementary material \ref{secsupplfig}), which shows that the structure retrieved by the LBM is clearly different from the one from $M$.
Even if the LBM estimates the correct number of groups, its ARI score is very low because the estimated structure is different. The \monmodele{} can have a better estimation of the clustering even if more than half of the interactions are missing. For LBM and CoOP-LBM, as $G$ is increasing, the ARI score increases because $\mathbb{P}(N_{i,j}=0)$ decreases. Additional results can be found in the supplementary material.

\subsubsection{Recovery of the support} 
The recovery of missing interactions is firstly assessed by calculating the Area Under the Curve of the ROC curve for missing interactions. The result will be compared with probabilities given by the LBM \citep{terry_finding_2020}.
We recall that with the \monmodele{}, the probability of having a missing an interaction 
when $R_{i,j}=0$ is \begin{equation}
\mathbb{P}(M_{i,j}=1|R_{i,j}=0,Z_{i}^1=k,Z^2_{j}=l)=\frac{\pi_{kl}e^{-\lambda_i  \mu_j G}}
{1-\pi_{kl}(1-e^{-\lambda_i\mu_j G})}\end{equation}
and can be estimated by replacing $\bZ^1,\bZ^2$ by the obtained clusterings $\widehat{\bZ}^1,\widehat{\bZ}^2$ and  $\bpi, \blambda,\bmu , G$ by their estimates.
This estimated probability tends to 0 when $\lambda_i\mu_j G$ increases, whereas it tends to $\pi_{kl}$ if $\lambda_i\mu_j G$ decreases. This latest limit is the estimator of the probability of a missing data 
given by \cite{terry_finding_2020} in the case of the classical LBM. For the DCSBM, we consider as an estimator $1-e^{-\Omega_{i,j}}$  which is the probability for Poisson random variable of parameter $\Omega_{i,j}$ to be greater or equal to $1$. In order to compare the results, we calculated the AUC of the ROC, which measures how much the models are capable of distinguishing between which interactions are missing and which are not. An AUC close to 1 describes a perfect separation, whereas an AUC close to 0.5 describes a model where the distinction is as good as a distinction made randomly.
\begin{figure}
\centering
\includegraphics[width=8.5cm]{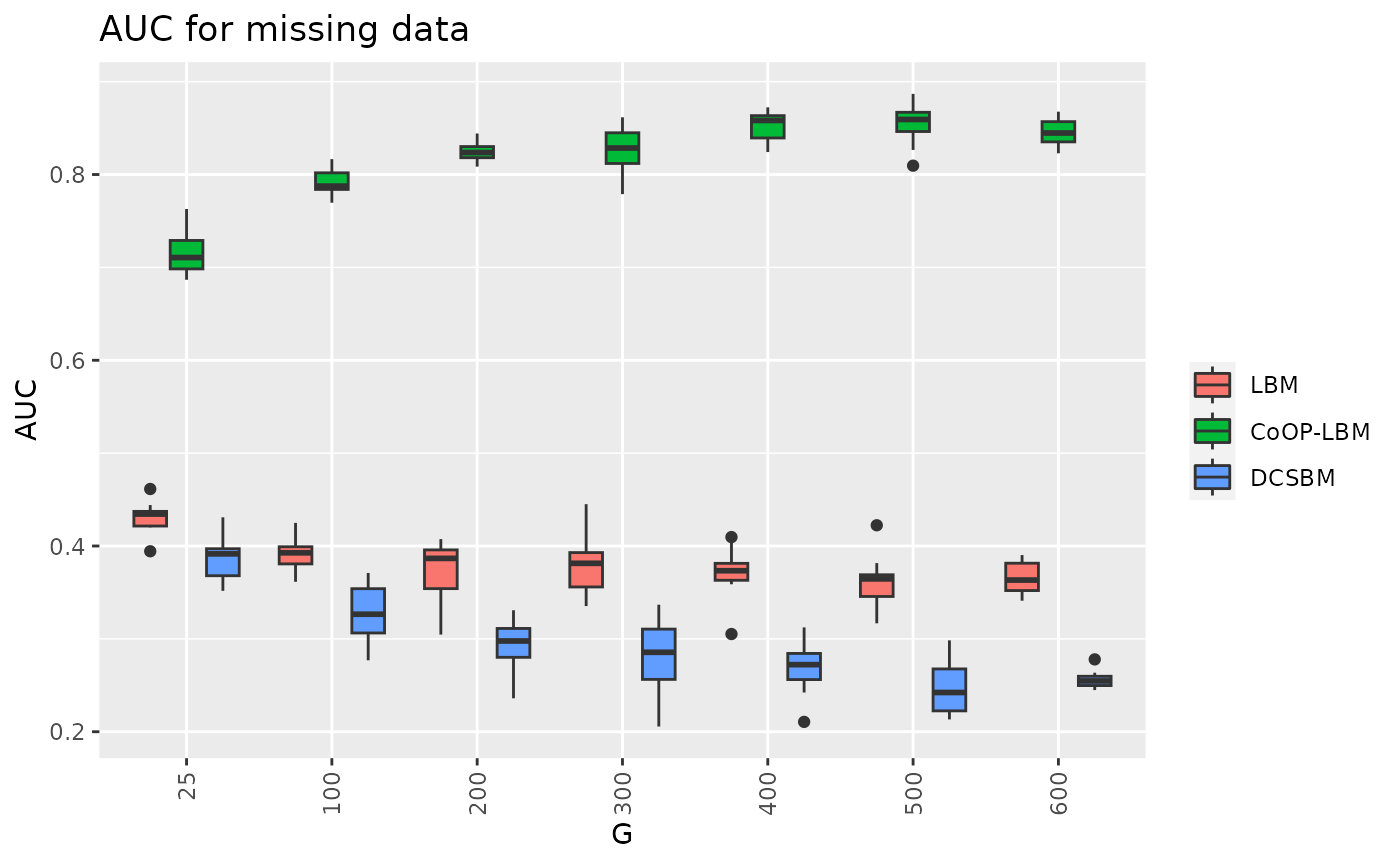}
\caption{AUC of the ROC curves for the prediction of missing data}
\label{fig6}
\end{figure}
It can be seen in
 \cref{fig6} that on average, the \monmodele{} gives a better estimation in our simulation than the LBM. The \monmodele{} even reached an AUC of 0.9 for $G=500$, and is on average equal to 0.824, whereas the maximum AUC reached by the LBM  and DCSBM was lower than 0.5. On \cref{fig7}, an example of the estimated probabilities can be seen. In the LBM that has been fit on the observed matrix, we can see large white bands in the middle where a lot of missing data has not been detected, while the same areas are reddish in the \monmodele{}. On the contrary, the area on the bottom right has more red in the LBM, compared to the  \monmodele{}, which is more accurate to identify where the interactions are missing or not.

\begin{figure}[H]
    \centering
    \includegraphics[width=0.8\textwidth]{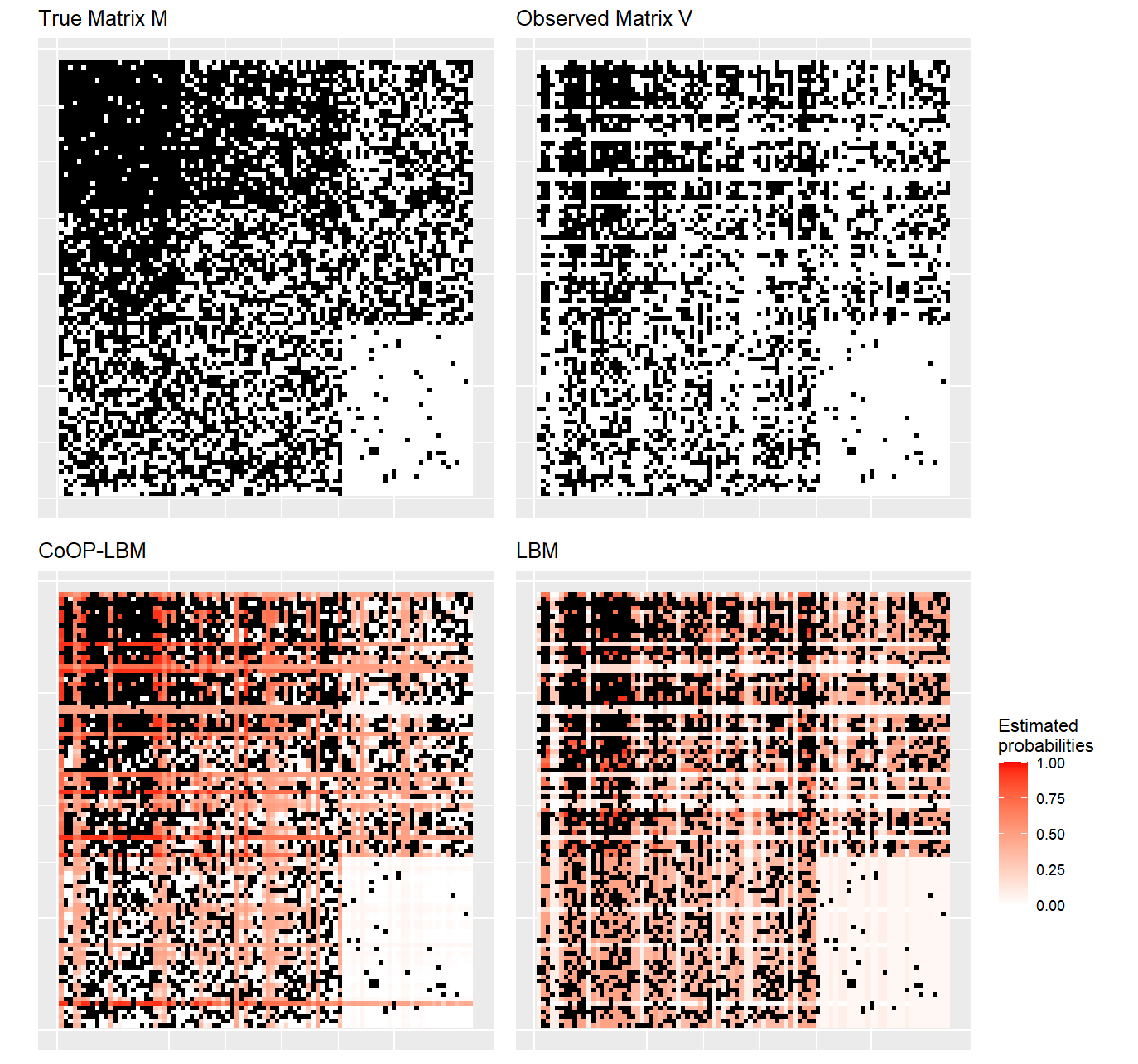}
    \caption{Estimated probabilities of missing interaction. On the top left, $M$ is the full matrix of interaction. On the top right,  $V$ is the observed matrix. On the bottom left, the estimated probabilities of missing interaction by CoOP-LBM are given with shades of red. On the bottom right, the estimated probabilities of missing interaction given by the LBM. All the matrices have been permuted according to the true clustering of $M$. Looking at the shades of red, we can see that CoOP-LBM has better retrieved the initial structure of the matrix than the LBM. }
    \label{fig7}
\end{figure}

\subsubsection{Sampling completeness and connectivity} 

In order to estimate a sampling completeness, we can use species richness estimator, such as the Chao 1 estimator \citep{chao_nonparametric_1984}. Its purpose is to estimate the total number of species in a population given the number of time each observed species has been counted. This methodology can be used to estimate the total number of interaction in the network, given the number of time each interaction has been observed. In our setting, we define the Chao 1 to estimate the total number of different interactions as 
\begin{equation}
\hat{S} = s + \frac{f_1(f_1-1)}{2(f_2+1)} \end{equation}
where $f_1$ is the number of interactions observed only once, $f_2$ is number of interactions observed only twice, and $s = \sum_{i,j}V_{i,j}$ is the total number of different interactions.  If $\hat{S}$ is close to $1$, then almost all the different interactions have been sampled, if $\hat{S}$ is close to $0$ then  a lot of interactions are missing in the sampling. It can be estimated thanks to the \texttt{vegan} R package. 
We can also use a coverage estimator \citep{chao_coverage-based_2012}, which represents the rate of total number of species in a population that belong to the species present in the sample. It can also be defined as the probability that the next species drawn has already been seen at least once before. Once again this methodology can be used to estimate the coverage of the network. This has been studied in \cite{terry_finding_2020} in order to estimate the number of missing interactions for each row of the observed matrix. To estimate the coverage at the whole network level, we will use
\begin{equation}
\hat{C}=1-\frac{f_1}{n}\frac{f_1(n-1)}{f_1(n-1)+2(f_2+1)}\end{equation} where $n = \sum_{i,j}R_{i,j}$.
 Knowing the coverage of the sampling can help estimate the true connectivity of the network. 

In this simulation study, we wish to compare the estimation given by the \monmodele{} and the estimation of total connectivity adjusted with the Chao 1 estimator, and the Chao coverage estimator and the Poisson degree-corrected SBM (DCSBM). We define the corrected connectivity estimators as

 \begin{align}
     &C_{Chao\; 1} = \frac{\hat{S}}{n_1 n_2}, \quad C_{Chao\; coverage} = \frac{n}{n_1 n_2}\times\frac{1}{\hat{C}},\quad C_{CoOP} = \frac{1}{n_1 n_2}\sum_{i=1}^{n_1}\sum_{j=1}^{n_2}\mathbb{P}(M_{i,j}=1|R_{i,j})
 \end{align}

 \begin{align}
     C_{DCSBM} =  \frac{1}{n_1 n_2}\sum_{i=1}^{n_1}\sum_{j=1}^{n_2}(1 - e^{-\Omega_{i,j}})
 \end{align}
  
\begin{figure}[H]
\centering
\includegraphics[width=0.65\textwidth]{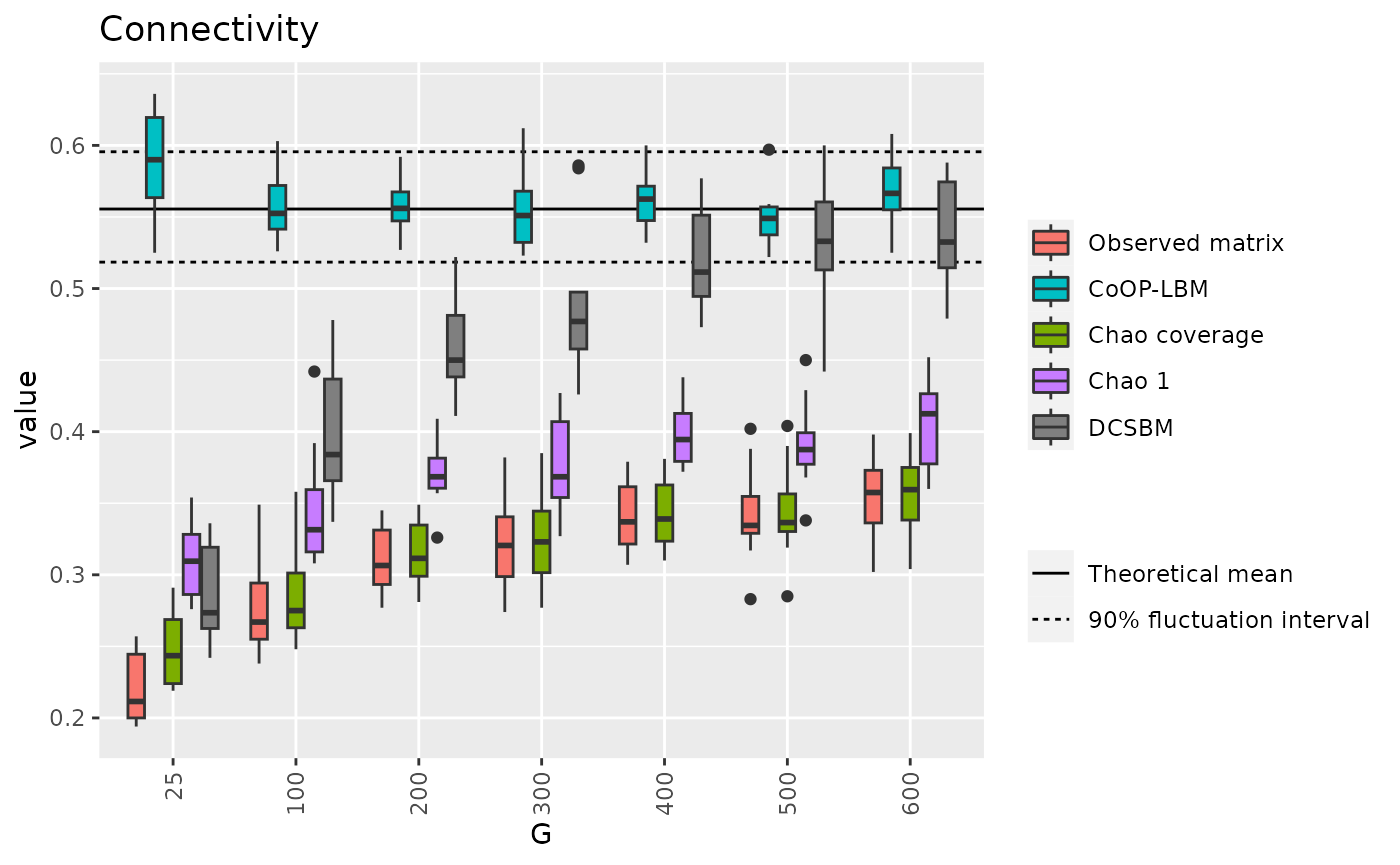}
\caption{Comparison of estimated connectivity on the observed matrix, and with the Chao-adjusted estimator, the CoOP-LBM estimator and the DCSBM}.
\label{fig11}
\end{figure}

As $M$ is an LBM, we can know its theoretical average connectivity which is given by $\alpha^\top \pi \beta$. In practice, the value of the connectivity fluctuate slightly around the average. In \cref{fig11} the average connectivity of $M$ is represented with a solid line while the dotted lines delimit the 0.05-quantile and 0.95-quantile of the different connectivity measured on $M$. We can see that \monmodele{} is much better to recover the connectivity compared to using both Chao estimators in this case. This is due to the fact that the size of the observed matrix $V$ does not increase as much with $G$ as the size of the sampling. Consequently, the Chao estimators quickly converge to $1$ with increasing $G$, and does not correct adequately  the connectivity. In this simulation setting, the DCSBM corrects adequatly the connectivity only for values of $G$ higher than 400.

\subsubsection{Nestedness, modularity}
Several metrics are used in the ecology community to describe networks, such as nestedness  \citep{terry_finding_2020,Lara-Romero-2019,de_manincor_how_2020,fortuna_nestedness_2010}, or modularity  \citep{terry_finding_2020,de_manincor_how_2020,guimera_missing_2009,fortuna_nestedness_2010}.

In the ecological terminology, a specialist is a species that has very few interactions, contrary to the generalist that has a lot more. 
The purpose of nestedness is to describe in an interaction network how much specialists interact with proper subsets of the species that generalists interact with.  Several metrics have been develop to define or estimate this nestedness. Here we will use the NODF metric from \cite{NODF}. A network is perfectly nested if its nestedness is close to $1$, and is not nested if its nestedness reaches 0. Another tool to study the structure of a network is the modularity, which measures the strength of division of a network into modules. Networks with high modularity have dense connections between the species within modules but sparse connections between species from different modules. Modularity also ranges between 0 and 1.

The nestedness and modularity will be applied on the initial matrix $M$ and on the observed matrix $V$. Then, the matrix $V$ will be completed with three different methods to see if an imputation can correct these metrics. Let $n_{miss}\in \mathbb{N}$.
\begin{enumerate}
    \item Among the coordinates $(i,j)$ where $V_{i,j}=0$, sample without replacement $n_{miss}$  coordinates with probability proportional to $\widehat{\pi_{kl}}$ (probability of connection in the cluster) which has been estimated by the LBM, or proportional to $1-e^{-\widehat{\Omega}_{i,j}} $ which has been estimated by the DCSBM. Substitute the $0$ at these sampled coordinates with $1$.

    \item For all $(i,j)$ such that $V_{i,j}=0$, draw a Bernoulli random variable of probability  ${\mathbb{P}(M_{i,j}=1|R_{i,j}=0)}$  which has been estimated by the \monmodele{}.

    \item  Among the coordinates $(i,j)$ where $V_{i,j}=0$, sample uniformly without replacement $n_{miss}$ coordinates. Substitute the $0$ at these sampled coordinates with $1$.
\end{enumerate}

$n_{miss}$ is an estimation of the number of interaction missed. It could have been provided by exterior experts or other similar studies. To give a slight advantage to method 1 and 3, we suppose that we have access to an exact estimation of ${n_{miss} = \sum_{i,j}M_{i,j} - \sum_{i,j}V_{i,j}}$, the corresponding estimator for the LBM will be called the oracle LBM. The nestedness (NODF) and modularity are both calculated with the \texttt{bipartite} package.

\begin{figure}[H]
\centering
\includegraphics[width=\textwidth]
{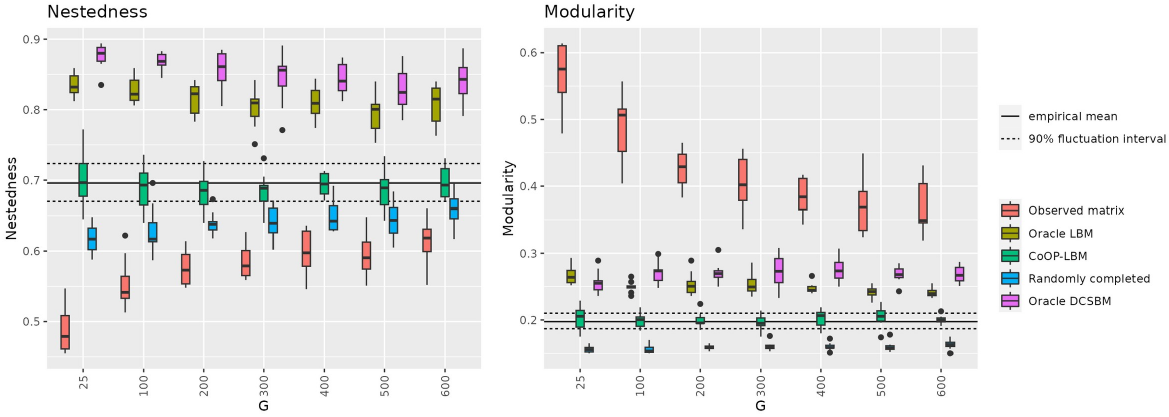}\caption{Estimation of Nestedness (left) and modularity (right)} 
\label{fig12}
\end{figure}

As we can see in \cref{fig12}, \monmodele{} is able to correct the estimation of nestedness and modularity, without having the precise number of missing interactions. Meanwhile, the LBM and the DCSBM perform even worse than the random completion to estimate 
the nestedness, and do no better 
than the random completion to estimate modularity.

\section{Application on real data}
\label{secappdata}
\subsection{Application on the network from Olesen \textit{et al.}, 2002}
  \begin{figure}[H]
    
    \begin{minipage}{.3\textwidth}
    \centering
    \includegraphics[width=1.2\textwidth]{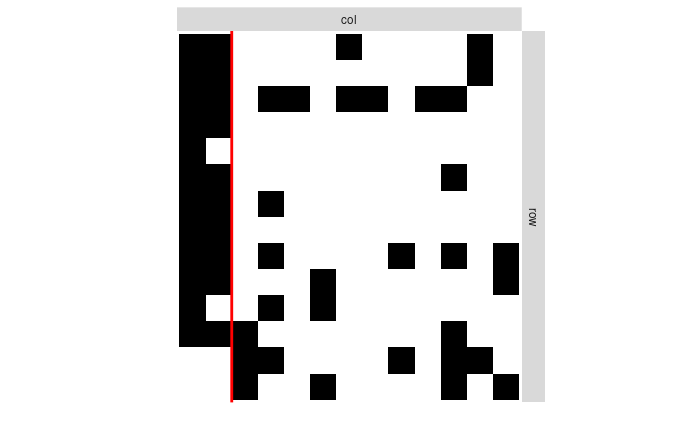}
    
    Estimated LBM
    \end{minipage}%
    \begin{minipage}{.3\textwidth}
    \centering
    \includegraphics[width=1.2\textwidth]{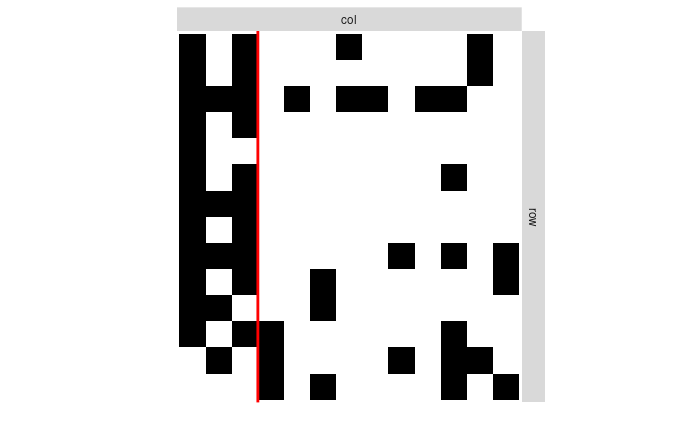}
    
    Estimated CoOP-LBM
    \end{minipage}
     \begin{minipage}{.3\textwidth}
    \centering
    \includegraphics[width=0.82\textwidth]{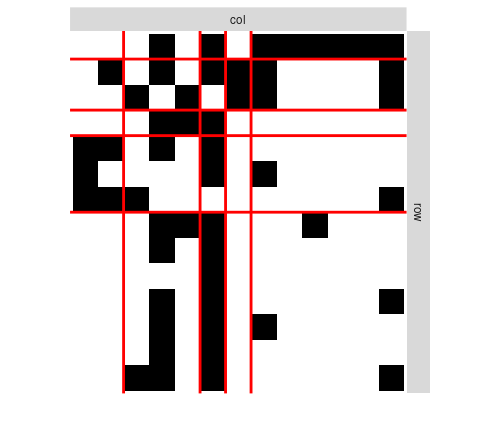}
    
    Estimated DCSBM
    \end{minipage}
     \caption{Difference between the estimated LBM and CoOP-LBM for the network of \cite{olesen_invasion_2002}, the only difference is observed for the insect species \textit{Lycaenidae pirithous}.} 
     \label{Olesen_graph}
    \end{figure}

\cite{dore_relative_2021} have provided a dataset of 299 networks, that have been aggregated into 123 networks according to location. Only 70 of these aggregated networks are weighted networks.
Looking at this set of 70 quantitative pollination networks, two networks from \cite{olesen_invasion_2002} are the best sampled according to the sampling effort criterion, which has been calculated as the number of person-hour spent sampling divided by the number of all possible interactions, and has a Chao coverage estimator equal to 0.995 which is very close to a perfect coverage. The network from \cite{olesen_invasion_2002} has been sampled in the Mauritian Ile aux Aigrettes, where 14 species of plants and 13 species of insect have been observed. The network contains 42 different interactions, sampled 1395 times in total, and is small enough to be easily interpreted.
\begin{figure}[H]
\centering
\includegraphics[width=9.5cm]{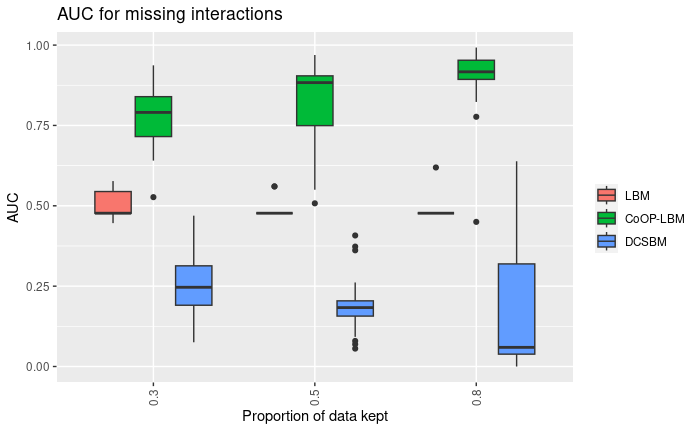}
\caption{AUC of the estimation of probabilities of missing interactions for the network of \cite{olesen_invasion_2002}. }
\label{fig14}
\end{figure}
LBM, DCSBM and CoOP-LBM have been fitted on this network. As we can see in \cref{Olesen_graph}, LBM and CoOP-LBM methods categorize the plants into the same group, and the insect into two groups, which can be interpreted as highly and weakly connected. Unlike the LBM, CoOP-LBM considers that the butterfly species \textit{Lycaenidae pirithous} is among the highly connected insects. This could be due to the fact that even if \textit{Lycaenidae pirithous} has been observed only 7 times, it was present on 5 different plant species.  According to the CoOP-LBM, only half of the interactions of \textit{Lycaenidae pirithous} have been sampled, whereas all the other species of plants and insects have a coverage close to 100\%. A longer observation of \textit{Lycaenidae pirithous} could lead to discovering new interactions, which could justify why \textit{Lycaenidae pirithous} could be in the highly connected group.

To validate the calculated probabilities and coverage, we decided to sub-sample all interactions as done in \cite{terry_finding_2020}, by considering the empirical interaction frequency to be the true network to define a multinomial distribution. We subsampled the empirical network 30 times, each time keeping between 30\% and 80\% of the original sample size, which is the number of total observations.

As we can see in \cref{fig14}, the LBM and the DCSBM struggles to achieve an AUC higher than 0.5 in this situation. This situation is similar to the simulation done in \cref{fig7} where the LBM considers that the probability of having missing interaction is higher in the group that has been observed as highly connected. 

\subsection{Application on the data set of Doré, Fontaine \& Thébault, 2021}

In the 70 weighted networks, 2256 different species of plants and 5665 species of insects are observed  in total. The matrices differ in size, from $10\times 12$ to $112 \times 845$. We fit the LBM, the CoOP-LBM, and the DCSBM on each of these 70 networks. For the DCSBM we limit the range exploration of models by constraining $Q_1+ Q_2 \leq 20$.  
It is difficult to validate whether our method can retrieve the missing interactions because the validation process should assume that the observed matrix before the sub-sampling is the full interaction matrix, which is probably not the case. The method will discover interactions that have not been observed and will penalize the AUC score because of the false positive. We show on \url{https://anakokemre.github.io/CoOP-LBM/articles/Validation_on_observed_networks.html} how the validation process can disadvantage the CoOP-LBM.
In the previous section, the studied network 
was obtained with a large sampling effort and was rather small so we can assume that it was close to a complete network which explains the good performance
of the missing interaction prediction.

\begin{figure}[b]
    \begin{minipage}{.33\textwidth}
    \centering
    \includegraphics[width=1\textwidth]{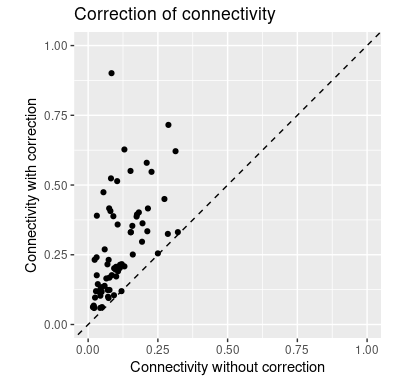}

    \end{minipage}%
    \begin{minipage}{.33\textwidth}
    \centering
    \includegraphics[width=1\textwidth]{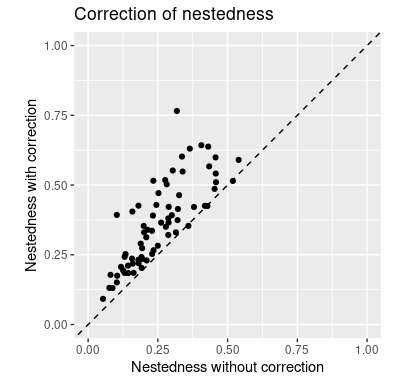}

    \end{minipage}
    \begin{minipage}{.33\textwidth}
    \centering
    \includegraphics[width=1\textwidth]{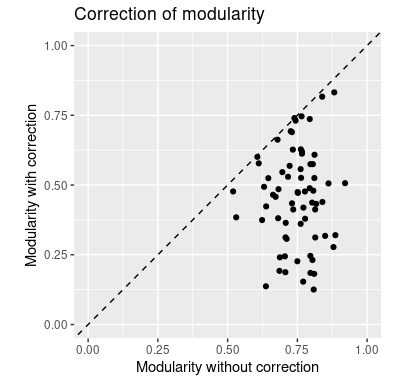}

    \end{minipage}
    
    \caption{Connectivity, nestedness and modularity on the observed matrix, and on the CoOP-LBM corrected matrix, estimated on the network set from \cite{dore_relative_2021}}

    \label{fig15}
    \end{figure}

% \PB{entree en matiere aride, dire
% we fit the cooplbm and the lbm on the 70 networks, je parlerai des corrections des metrics apres avoir comparer les nb de groupes trouvés et les auc}
%     We correct the estimation of metrics as we have done in the previous part, by completing the matrix of interaction with Bernoulli's variable with probabilities ${\mathbb{P}(M_{i,j}=1|R_{i,j}=0)}$.
%   As we can see in \cref{fig15}, correcting the estimation increases the connectivity and the nestedness, but decreases the modularity. It also decreases the number of estimated groups (table \ref{table2}) as expected, because the CoOP-LBM will merge the groups originated by the sampling.
The main difference between the estimated structure of CoOP-LBM compared to the other method is that the number of estimated groups by CoOP-LBM is lower. Compared to the LBM (table \ref{table2}) %as expected
, it  will merge the groups originated by the sampling. It is also much lower than the number of estimated groups by the Poisson DCSBM (table \ref{table3}), which clusters the nodes according to the expected value of the Poisson parameter. As the number of observations can variate greatly, notably because of abundance or detectability, the DCSBM tends to have a fit a model with a great number of groups. 

\begin{table}[H]
\caption{Estimated number of groups for plants (left) and insects (right) on set of networks from \cite{dore_relative_2021} estimated by CoOP-LBM and LBM. For example, the $18$ in the left table should be read as "there are 18 networks where the LBM estimated the number of groups for plants as 2, while the CoOP-LBM estimated it as 1".\label{table2}
}

\scalebox{0.8}{
\begin{tabular}{cc|ccccc|c}
\cline{3-7}
\textbf{} &
  &
  \multicolumn{5}{c|}{LBM} &
  \\ \cline{2-8} 
 & \multicolumn{1}{|c|}{\cellcolor[HTML]{FFFFC7}n of groups} &
  \multicolumn{1}{c|}{\cellcolor[HTML]{FFFFC7}1} &
  \multicolumn{1}{c|}{\cellcolor[HTML]{FFFFC7}2} &
  \multicolumn{1}{c|}{\cellcolor[HTML]{FFFFC7}3} &
  \multicolumn{1}{c|}{\cellcolor[HTML]{FFFFC7}4} &
  \cellcolor[HTML]{FFFFC7}5 &
  \multicolumn{1}{c|}{Total} \\ \hline
\multicolumn{1}{|c|}{} &
  \cellcolor[HTML]{FFFFC7}{\color[HTML]{333333} 1} &
  \multicolumn{1}{c|}{2} &
  \multicolumn{1}{c|}{18} &
  \multicolumn{1}{c|}{9} &
  \multicolumn{1}{c|}{1} &
  0 &
  \multicolumn{1}{c|}{30} \\ \cline{2-8} 
\multicolumn{1}{|c|}{} &
  \cellcolor[HTML]{FFFFC7}{\color[HTML]{333333} 2} &
  \multicolumn{1}{c|}{0} &
  \multicolumn{1}{c|}{16} &
  \multicolumn{1}{c|}{14} &
  \multicolumn{1}{c|}{3} &
  0 &
  \multicolumn{1}{c|}{33} \\ \cline{2-8} 
\multicolumn{1}{|c|}{} &
  \cellcolor[HTML]{FFFFC7}{\color[HTML]{333333} 3} &
  \multicolumn{1}{c|}{0} &
  \multicolumn{1}{c|}{0} &
  \multicolumn{1}{c|}{3} &
  \multicolumn{1}{c|}{2} &
  0 &
  \multicolumn{1}{c|}{5} \\ \cline{2-8} 
\multicolumn{1}{|c|}{} &
  \cellcolor[HTML]{FFFFC7}{\color[HTML]{333333} 4} &
  \multicolumn{1}{c|}{0} &
  \multicolumn{1}{c|}{0} &
  \multicolumn{1}{c|}{0} &
  \multicolumn{1}{c|}{0} &
  0 &
  \multicolumn{1}{c|}{0} \\ \cline{2-8} 
\multicolumn{1}{|c|}{\multirow{-5}{*}{\parbox{1cm}{CoOP-LBM}}} &
  \cellcolor[HTML]{FFFFC7}{\color[HTML]{333333} 5} &
  \multicolumn{1}{c|}{0} &
  \multicolumn{1}{c|}{0} &
  \multicolumn{1}{c|}{0} &
  \multicolumn{1}{c|}{1} &
  1 &
  \multicolumn{1}{c|}{2} \\ \hline
\multicolumn{1}{c|}{} &
  Total &
  \multicolumn{1}{c|}{2} &
  \multicolumn{1}{c|}{34} &
  \multicolumn{1}{c|}{26} &
  \multicolumn{1}{c|}{7} &
  1 &
  \multicolumn{1}{c|}{} \\ \cline{2-8} 
\end{tabular}
}
\scalebox{0.8}{
\begin{tabular}{cc|ccccc|c}
\cline{3-7}
\textbf{} &
  &
  \multicolumn{5}{c|}{LBM} &
  \\ \cline{2-8} 
 &
  \multicolumn{1}{|c|}{\cellcolor[HTML]{FFFFC7}n of groups}&
  \multicolumn{1}{c|}{\cellcolor[HTML]{FFFFC7}1} &
  \multicolumn{1}{c|}{\cellcolor[HTML]{FFFFC7}2} &
  \multicolumn{1}{c|}{\cellcolor[HTML]{FFFFC7}3} &
  \multicolumn{1}{c|}{\cellcolor[HTML]{FFFFC7}4} &
  \cellcolor[HTML]{FFFFC7}5 &
  \multicolumn{1}{c|}{Total} \\ \hline
\multicolumn{1}{|c|}{} &
  \cellcolor[HTML]{FFFFC7}{\color[HTML]{333333} 1} &
  \multicolumn{1}{c|}{18} &
  \multicolumn{1}{c|}{24} &
  \multicolumn{1}{c|}{5} &
  \multicolumn{1}{c|}{0} &
  1 &
  \multicolumn{1}{c|}{48} \\ \cline{2-8} 
\multicolumn{1}{|c|}{\parbox{1cm}{CoOP-LBM}} &
  \cellcolor[HTML]{FFFFC7}{\color[HTML]{333333} 2} &
  \multicolumn{1}{c|}{0} &
  \multicolumn{1}{c|}{15} &
  \multicolumn{1}{c|}{5} &
  \multicolumn{1}{c|}{1} &
  0 &
  \multicolumn{1}{c|}{21} \\ \cline{2-8} 
\multicolumn{1}{|c|}{} &
  \cellcolor[HTML]{FFFFC7}{\color[HTML]{333333} 3} &
  \multicolumn{1}{c|}{0} &
  \multicolumn{1}{c|}{1} &
  \multicolumn{1}{c|}{0} &
  \multicolumn{1}{c|}{0} &
  0 &
  \multicolumn{1}{c|}{1} \\ \hline
\multicolumn{1}{c|}{} &
  Total &
  \multicolumn{1}{c|}{18} &
  \multicolumn{1}{c|}{40} &
  \multicolumn{1}{c|}{10} &
  \multicolumn{1}{c|}{1} &
  1 &
  \multicolumn{1}{c|}{} \\ \cline{2-8} 
\end{tabular}}
\end{table}

\begin{table}[H]
\caption{Estimated number of groups for plants (left) and insects (right) on set of networks from \cite{dore_relative_2021} estimated by CoOP-LBM and DCSBM. For example, the $9$ in the left table should be read as "there are 9 networks where the DCSBM estimated the number of groups for plants as 3, while the CoOP-LBM estimated it as 1".}\label{table3}
\scalebox{0.6}{
\begin{tabular}{cc|cccccccccc|c}
\cline{3-12}
 &
   &
  \multicolumn{10}{c|}{DCSBM} &
   \\ \cline{2-13} 
\multicolumn{1}{c|}{} &
  \cellcolor[HTML]{FFFFC7}n of groups &
  \multicolumn{1}{c|}{\cellcolor[HTML]{FFFFC7}1} &
  \multicolumn{1}{c|}{\cellcolor[HTML]{FFFFC7}2} &
  \multicolumn{1}{c|}{\cellcolor[HTML]{FFFFC7}3} &
  \multicolumn{1}{c|}{\cellcolor[HTML]{FFFFC7}4} &
  \multicolumn{1}{c|}{\cellcolor[HTML]{FFFFC7}5} &
  \multicolumn{1}{c|}{\cellcolor[HTML]{FFFFC7}6} &
  \multicolumn{1}{c|}{\cellcolor[HTML]{FFFFC7}7} &
  \multicolumn{1}{c|}{\cellcolor[HTML]{FFFFC7}8} &
  \multicolumn{1}{c|}{\cellcolor[HTML]{FFFFC7}10} &
  \cellcolor[HTML]{FFFFC7}13 & \multicolumn{1}{c|}{Total}
   \\ \hline
\multicolumn{1}{|c|}{\cellcolor[HTML]{FFFFFF}} &
  \cellcolor[HTML]{FFFFC7}1 &
  \multicolumn{1}{c|}{10} &
  \multicolumn{1}{c|}{4} &
  \multicolumn{1}{c|}{9} &
  \multicolumn{1}{c|}{2} &
  \multicolumn{1}{c|}{3} &
  \multicolumn{1}{c|}{0} &
  \multicolumn{1}{c|}{2} &
  \multicolumn{1}{c|}{0} &
  \multicolumn{1}{c|}{0} &
  0 &
   \multicolumn{1}{c|}{30} \\ \cline{2-13} 
\multicolumn{1}{|c|}{\cellcolor[HTML]{FFFFFF}} &
  \cellcolor[HTML]{FFFFC7}2 &
  \multicolumn{1}{c|}{4} &
  \multicolumn{1}{c|}{4} &
  \multicolumn{1}{c|}{2} &
  \multicolumn{1}{c|}{2} &
  \multicolumn{1}{c|}{5} &
  \multicolumn{1}{c|}{5} &
  \multicolumn{1}{c|}{5} &
  \multicolumn{1}{c|}{5} &
  \multicolumn{1}{c|}{0} &
  1 &
   \multicolumn{1}{c|}{33} \\ \cline{2-13} 
\multicolumn{1}{|c|}{\cellcolor[HTML]{FFFFFF}} &
  \cellcolor[HTML]{FFFFC7}3 &
  \multicolumn{1}{c|}{2} &
  \multicolumn{1}{c|}{0} &
  \multicolumn{1}{c|}{0} &
  \multicolumn{1}{c|}{0} &
  \multicolumn{1}{c|}{1} &
  \multicolumn{1}{c|}{1} &
  \multicolumn{1}{c|}{1} &
  \multicolumn{1}{c|}{0} &
  \multicolumn{1}{c|}{0} &
  0 &
   \multicolumn{1}{c|}{5} \\ \cline{2-13} 
\multicolumn{1}{|c|}{\cellcolor[HTML]{FFFFFF}} &
  \cellcolor[HTML]{FFFFC7}4 &
  \multicolumn{1}{c|}{0} &
  \multicolumn{1}{c|}{0} &
  \multicolumn{1}{c|}{0} &
  \multicolumn{1}{c|}{0} &
  \multicolumn{1}{c|}{0} &
  \multicolumn{1}{c|}{0} &
  \multicolumn{1}{c|}{0} &
  \multicolumn{1}{c|}{0} &
  \multicolumn{1}{c|}{0} &
  0 &
   \multicolumn{1}{c|}{0} \\ \cline{2-13} 
\multicolumn{1}{|c|}{\multirow{-5}{*}{\cellcolor[HTML]{FFFFFF}\begin{tabular}[c]{@{}c@{}}CoOP-\\ LBM\end{tabular}}} &
  \cellcolor[HTML]{FFFFC7}5 &
  \multicolumn{1}{c|}{0} &
  \multicolumn{1}{c|}{0} &
  \multicolumn{1}{c|}{0} &
  \multicolumn{1}{c|}{0} &
  \multicolumn{1}{c|}{0} &
  \multicolumn{1}{c|}{0} &
  \multicolumn{1}{c|}{0} &
  \multicolumn{1}{c|}{1} &
  \multicolumn{1}{c|}{1} &
  0 &
\multicolumn{1}{c|}{2} \\ \hline
\multicolumn{1}{c|}{} &
  Total &
  \multicolumn{1}{c|}{16} &
  \multicolumn{1}{c|}{8} &
  \multicolumn{1}{c|}{11} &
  \multicolumn{1}{c|}{4} &
  \multicolumn{1}{c|}{9} &
  \multicolumn{1}{c|}{6} &
  \multicolumn{1}{c|}{8} &
  \multicolumn{1}{c|}{6} &
  \multicolumn{1}{c|}{1} &
  1 &
   \multicolumn{1}{c|}{} \\ \cline{2-13} 
\end{tabular}
}
\scalebox{0.6}{
\begin{tabular}{cc|ccccccccclc|c}
\cline{3-13}
 &
   &
  \multicolumn{11}{c|}{DCSBM} &
   \\ \cline{2-14} 
\multicolumn{1}{c|}{} &
  \cellcolor[HTML]{FFFFC7}n of groups &
  \multicolumn{1}{c|}{\cellcolor[HTML]{FFFFC7}1} &
  \multicolumn{1}{c|}{\cellcolor[HTML]{FFFFC7}2} &
  \multicolumn{1}{c|}{\cellcolor[HTML]{FFFFC7}3} &
  \multicolumn{1}{c|}{\cellcolor[HTML]{FFFFC7}4} &
  \multicolumn{1}{c|}{\cellcolor[HTML]{FFFFC7}5} &
  \multicolumn{1}{c|}{\cellcolor[HTML]{FFFFC7}6} &
  \multicolumn{1}{c|}{\cellcolor[HTML]{FFFFC7}7} &
  \multicolumn{1}{c|}{\cellcolor[HTML]{FFFFC7}8} &
  \multicolumn{1}{c|}{\cellcolor[HTML]{FFFFC7}9} &
  \multicolumn{1}{l|}{\cellcolor[HTML]{FFFFC7}10} &
  \cellcolor[HTML]{FFFFC7}12 &
  \multicolumn{1}{c|}{Total} \\ \hline
\multicolumn{1}{|c|}{\cellcolor[HTML]{FFFFFF}} &
  \cellcolor[HTML]{FFFFC7}1 &
  \multicolumn{1}{c|}{\cellcolor[HTML]{FFFFFF}14} &
  \multicolumn{1}{c|}{7} &
  \multicolumn{1}{c|}{10} &
  \multicolumn{1}{c|}{4} &
  \multicolumn{1}{c|}{4} &
  \multicolumn{1}{c|}{4} &
  \multicolumn{1}{c|}{2} &
  \multicolumn{1}{c|}{1} &
  \multicolumn{1}{c|}{2} &
  \multicolumn{1}{l|}{0} &
  0 &
  \multicolumn{1}{c|}{48} \\ \cline{2-14} 
\multicolumn{1}{|c|}{\cellcolor[HTML]{FFFFFF}} &
  \cellcolor[HTML]{FFFFC7}2 &
  \multicolumn{1}{c|}{2} &
  \multicolumn{1}{c|}{1} &
  \multicolumn{1}{c|}{1} &
  \multicolumn{1}{c|}{1} &
  \multicolumn{1}{c|}{1} &
  \multicolumn{1}{c|}{7} &
  \multicolumn{1}{c|}{2} &
  \multicolumn{1}{c|}{1} &
  \multicolumn{1}{c|}{1} &
  \multicolumn{1}{l|}{2} &
  2 &
  \multicolumn{1}{c|}{21} \\ \cline{2-14} 
\multicolumn{1}{|c|}{\multirow{-3}{*}{\cellcolor[HTML]{FFFFFF}\begin{tabular}[c]{@{}c@{}}CoOP-\\ LBM\end{tabular}}} &
  \cellcolor[HTML]{FFFFC7}3 &
  \multicolumn{1}{c|}{0} &
  \multicolumn{1}{c|}{0} &
  \multicolumn{1}{c|}{0} &
  \multicolumn{1}{c|}{0} &
  \multicolumn{1}{c|}{1} &
  \multicolumn{1}{c|}{0} &
  \multicolumn{1}{c|}{0} &
  \multicolumn{1}{c|}{0} &
  \multicolumn{1}{c|}{0} &
  \multicolumn{1}{l|}{0} &
  0 &
  \multicolumn{1}{c|}{1} \\ \hline
\multicolumn{1}{c|}{} &
  Total &
  \multicolumn{1}{c|}{16} &
  \multicolumn{1}{c|}{8} &
  \multicolumn{1}{c|}{11} &
  \multicolumn{1}{c|}{5} &
  \multicolumn{1}{c|}{6} &
  \multicolumn{1}{c|}{11} &
  \multicolumn{1}{c|}{4} &
  \multicolumn{1}{c|}{2} &
  \multicolumn{1}{c|}{3} &
  \multicolumn{1}{l|}{2} &
  2 &
  \multicolumn{1}{c|}{} \\ \cline{2-14} 
\end{tabular}
}
\end{table}

As the structure appears different when we fit CoOP-LBM, metrics describing the structure of the network will need a correction. We correct the estimation of metrics as we have done in the previous part, by completing the matrix of interaction with Bernoulli's variable with probabilities ${\mathbb{P}(M_{i,j}=1|R_{i,j}=0)}$.
 As we can see in \cref{fig15}, correcting the estimation increases the connectivity and the nestedness, but decreases the modularity.

\section{Discussion}
In this study, we presented CoOP-LBM, a model able to take into account the counting data of species interactions in order to correct the estimation of the structure of an ecological network. We also provide proof of identifiability of this model, and an original algorithm to estimate the parameters by simulating the missing interactions in the network. Our method performs better than the LBM on simulated data and on real data for finding the structure of an under-sampled network. Most of the time, CoOP-LBM estimated less groups for both insects and plants than the LBM or the DCSBM. This could suggest that a lot of groups observed in ecological networks could be a consequence of the sampling itself.  Further experiments have been performed on simulated data to test the robustness of our algorithm in different settings, for example with negative binomial distribution instead of Poisson distribution, or cases where the sampling efforts $(\lambda_i,\mu_j,G)$ may depend on some latent blocks. These experiments are available on the example vignettes of the R package CoOP-LBM at \url{https://anakokemre.github.io/CoOP-LBM/index.html}, and show that our algorithm is still able to recover well the true underlying structure.

As a byproduct, our methodology calculates the probability of having a missing interaction when no interaction has been observed. This allows correcting the connectivity of the model, and simulating the missing links to correct other metrics. Our method shows higher AUC values than the LBM and the DCSBM on the network from \cite{olesen_invasion_2002}. 
To further confirm the efficiency of our algorithm, the missing interactions estimated as highly probable by our method could be confirmed in laboratory tests, as done with genetic interaction networks. 

However, there is still room for improvement: for example, recording the time at which interactions are observed may be helpful to determine the sampling progress. Timed observations can be used to build an accumulation curve of new interactions, which should allow estimating the asymptotic number of interactions that could be observed. This estimation could provide additional information on the number of missing interactions in the network. Recording the sampling effort per abundance could provide additional information on the number of missing interactions in the network.
This study also neglects the effect of preferences and supposes that the counting data only comes from a product rule between the sampling efforts of pollinators and plants. Traits such as flower shape or insect body mass, which are important for plant-pollinator interactions \citep{chamberlain2014traits}, could also be used to further improve the estimation of clusters by using them as covariates in the model. These covariates could influence both the probabilities of connection and the Poisson distribution of frequencies. Another possibility to improve the algorithm is to take into account phylogeny. Two closely related species could have a higher probability of sharing the same interactions \citep{chamberlain2014traits}, or to be in the same group. All this exterior knowledge from experts could lead to a better reconstruction of the network. 

\section{Acknowledgments}
This work was partially supported
by the grant ANR-18-CE02-0010-01 of the French National Research Agency ANR (project EcoNet).
Emre Anakok was funded by the MathNum department of INRAE.
The authors thank Thomas Cortier whose Master thesis  was a preliminary exploration for this article.

\section{Funding}
This work was partially supported
by the grant ANR-18-CE02-0010-01 of the French National Research Agency ANR (project EcoNet).
Emre Anakok was funded by the MathNum department of INRAE.

\bibliography{biblio.bib}
\clearpage

\appendix
\section{Proof of CoOP-LBM identifiability}
\label{secproof}

\input{identifiability}

\section{Details on the EM algorithm}
\label{M-Steb b}
\paragraph{M-Step b)}

The likelihood of $R$ given $M$ can be written as 

$$ \mathcal{L}(R|M) = \prod_{i,j}\left( e^{-\lambda_i \mu_j G} \frac{(\lambda_i \mu_j G)^{R_{i,j}}}{R_{i,j}!}\right)^{m_{i,j}} \underbrace{\left( \mathbbm{1}_{\{R_{i,j}=0\}}\right)^{1-m_{i,j}}}_{=1}.$$
The right part of the product is always equal to one because $m_{i,j}=0$ implies that $R_{i,j}=0$. Then, 
$$\log \mathcal{L}(R|M) = \sum_{i,j} m_{i,j} \left[ -\lambda_i \mu_j G  + R_{i,j}(\log(\lambda_i)+\log(\mu_j )+\log(G)) - \log(R_{i,j}!) \right] .$$

The derivative of this likelihood with respect to $k^{th}$ coordinate of $\lambda$ gives:
$$\frac{\partial\log\mathcal{L}}{\partial\lambda_k}= -\sum_{j=1}^{n_2}m_{k,j}\mu_{j}G+ \sum_{j=1}^{n_2}m_{k,j}\frac{R_{k,j}}{\lambda_k}\quad \text{thus} \quad
\frac{\partial\log\mathcal{L}}{\partial\lambda_k} = 0 \text{ if } \lambda_k= \frac{\sum_{j=1}^{n_2}m_{k,j}R_{k,j}}{\sum_{j=1}^{n_2}m_{k,j}\mu_{j}G}\,.$$ 

Note that because $m_{i,j}=0$ implies $R_{i,j}=0$ we have that $\sum_{j=1}^{n_2}m_{k,j}R_{k,j}= \sum_{j=1}^{n_2}R_{k,j} $. With the same reasoning, $$\frac{\partial\log\mathcal{L}}{\partial\mu_l} = 0 \text{ if } \mu_l= \frac{\sum_{i=1}^{n_1}R_{i,l}}{\sum_{i=1}^{n_1}m_{i,l}\lambda_{l}G}.$$

By replacing $\mu_j$ in the first derivative by the expression found in the second, we have 

$$\frac{\partial\log\mathcal{L}}{\partial\lambda_k} = 0 \text{ if } \lambda_k= \frac{\sum_{j=1}^{n_2}R_{k,j}}{\sum_{j=1}^{n_2}m_{k,j}\frac{\sum_{i=1}^{n_1}R_{i,j}}{\sum_{i=1}^{n_1}m_{i,j}\lambda_{j}G}G} = \frac{\sum_{j=1}^{n_2}R_{k,j}}{\sum_{j=1}^{n_2}m_{k,j}\frac{\sum_{i=1}^{n_1}R_{i,j}}{\sum_{i=1}^{n_1}m_{i,j}\lambda_{j}}}.$$ We can then find $\lambda$ with a fixed point iteration. However the fixed point of this equation is not unique, indeed if $\lambda$ is a fixed point, then $a \lambda$ with $a>0$ is also a fixed point. Because it is supposed that $max_i \lambda_i =1$, a solution $\lambda$ that has been found by the fixed point algorithm can be divided by its maximum coordinate in order to have an estimation of $\lambda$. Once $\lambda$ calculated, we know that $G\mu_l= \frac{\sum_{i=1}^{n_1}R_{i,l}}{\sum_{i=1}^{n_1}m_{i,l}\lambda_{l}}$ can be obtained. Finally, we deduce $G$ and $\mu$ by fixing $\max_j \mu_j = 1$.
As the matrix $M$ is not observed, it is replaced by a simulated version of it  $\Tilde{M}_{(n)}$.

\section{Proof of ICL formula}
\label{ICL}
The proof is similar to the proof given in \cite{daudin_mixture_2008}, the prior distribution $g$ can be decomposed as $g(\btheta|m_{Q_1,Q2}) = g(\balpha|m_{Q_1,Q2}) \times g(\bbeta|m_{Q_1,Q2} )\times g(\bpi,\blambda,\bmu,G|m_{Q_1,Q2})$, which allows us to write 
$$\log\mathcal{L}(R,\bZ^1, \bZ^2,m_{Q_1,Q2}) = \log\mathcal{L}(\bZ^1|m_{Q_1,Q2}) + \log\mathcal{L}( \bZ^2|m_{Q_1,Q2}) + \log\mathcal{L}(R|\bZ^1, \bZ^2,m_{Q_1,Q2}) .$$

Using a non-informative Jeffreys prior and by replacing the missing data $\bZ^1$ and $\bZ^2$ by $\widehat{\bZ^1} $ and $\widehat{\bZ^2} $, we can approximate the first two terms by \begin{align}
    &\log\mathcal{L}(Z^1|m_{Q_1,Q2}) + \log\mathcal{L}( \bZ^2|m_{Q_1,Q2})\approx\notag\\  &\underset{\balpha}{\max}\log\mathcal{L}(\widehat{\bZ^1}|\balpha,m_{Q_1,Q2})-\frac{Q_1-1}{2}\log(n_1)  + \underset{\bbeta}{\max}\log\mathcal{L}(\widehat{\bZ^2}|\beta,m_{Q_1,Q2}) - \frac{Q_2-1}{2}\log(n_2).
\end{align}

For the third term, we can consider that $R$ is made of $n_1\times n_2$ random variables, depending on labels and on $\bpi,\blambda,\bmu,G$. $\log\mathcal{L}(R|\bZ^1, \bZ^2, m_{Q_1,Q2})$ can be calculated using a BIC approximation:

\begin{align}
    &\log\mathcal{L}(R|\bZ^1, \bZ^2,m_{Q_1,Q2})\approx\notag\\ &\underset{\bpi,\blambda,\bmu,G}{max}\log\mathcal{L}(R|\widehat{\bZ^1}, \widehat{\bZ^2},\bpi,\blambda,\bmu,G,m_{Q_1,Q2}) - \frac{Q_1Q_2 + n_1 +n_2-1}{2}\log(n_1n_2).
\end{align}

The sum of all the terms yields the final result.

\section{Supplementary figures}
\label{secsupplfig}

\begin{figure}[H]
\centering
\includegraphics[width=10cm]{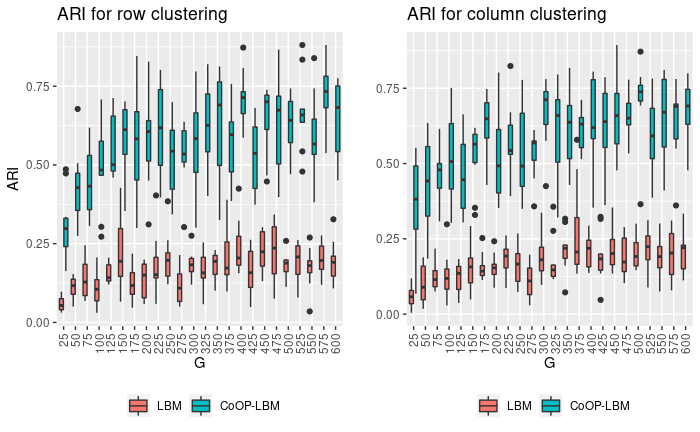}
\caption{ARI scores for rows and columns blocks when the number of blocks is known.}
\label{fig5.2}
\end{figure}

\begin{figure}[H]
  \centering
  \includegraphics[width=.4\textwidth]{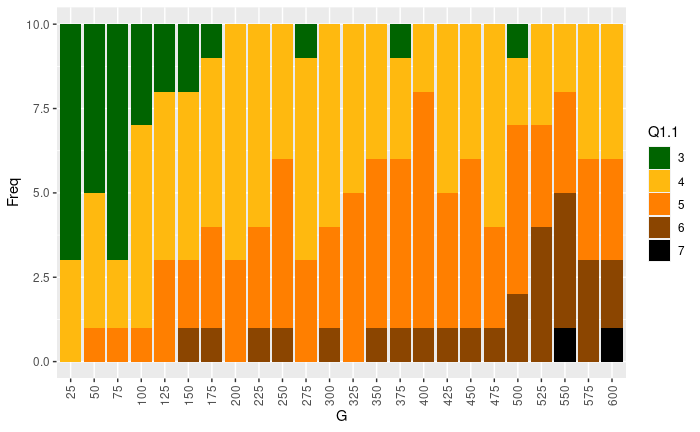}
  \includegraphics[width=.4\textwidth]{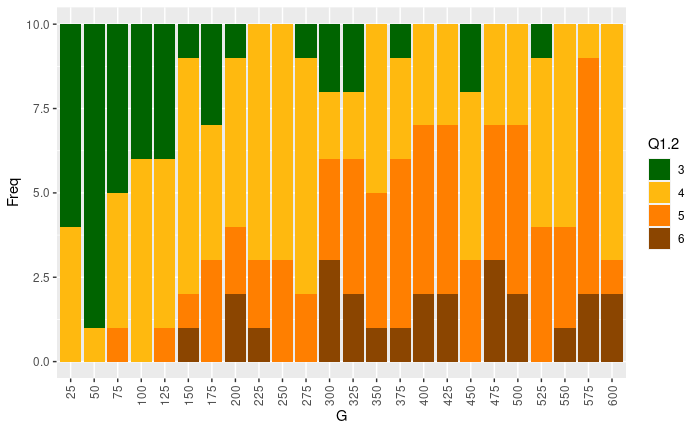}
  \label{fig_Q0}
  \caption{Estimated $Q_1$ (left) and $Q_2$ (right) with LBM for different value of $G$}
  \end{figure}
  \begin{figure}[H]
  \includegraphics[width=.4\textwidth]{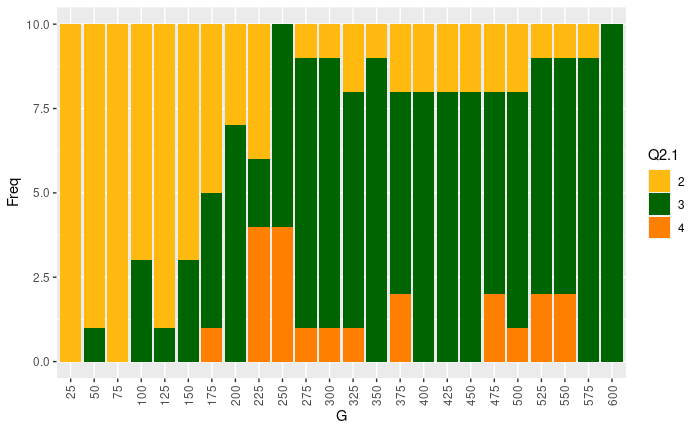}\quad
  \includegraphics[width=.4\textwidth]{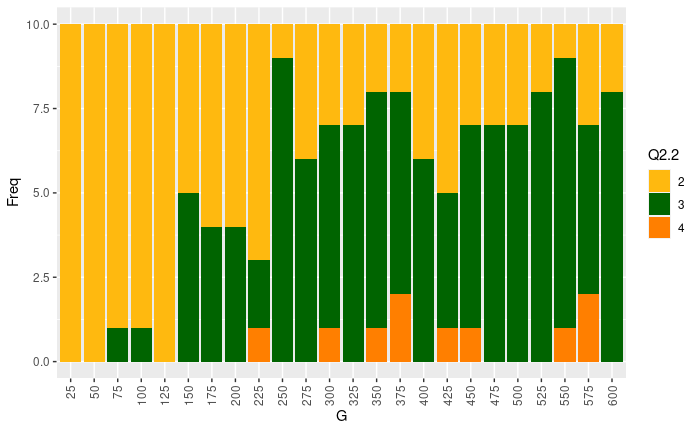}
  \caption{Estimated $Q_1$ (left) and $Q_2$ (right) with CoOP-LBM for different value of $G$}
  \label{fig_Q}
\end{figure}

\end{document}

%% file: identifiability.tex
\subsection{Introduction}

Let $R = (R_{i,j} , i = 1, \dots, n_1; j = 1, \dots, n_2)$ be the data matrix where $R_{i,j}\in \mathbb{N}$ is a non-negative number. It is supposed that $R = M \odot N $ with $M$ being a binary LBM and $N$ follows a Poisson distribution independant from $M$. More precisely, it is assumed that there exists a partition  $Z^1 = (Z^1_{i,k} ;\quad i = 1,\dots, n_1;\quad k \in \{1,\dots,Q_1\})$ and $Z^2 = (Z^2_{jl}; \quad j = 1,\dots, n_2;\quad l \in \{1,\dots,Q_2\})$ being binary indicators of row $i$ (resp. column $j$) belonging to row cluster $k$ (resp. column
cluster $j$), such that the random variable $M_{i,j}$ are conditionally independent knowing $Z^1$ and $Z^2$
with parameterized density 

$$\mathbb{P}(M_{i,j} = m | Z^1_{ik}=1, Z^2_{jl}=1) = \pi_{kl}^m (1-\pi_{kl})^{1-m}.$$

The conditional density of $M$ knowing  $Z^1$  and $Z^2$ is 

$$f_M (m|Z^1,Z^2;\theta_M) = \prod_{k,l} \prod_{i,j} \left( \pi_{kl}^{m_{i,j}}(1-\pi_{kl})^{1-m_{i,j}} \right)^{Z^1_{ik} Z^2_{jl}}.$$

It is also supposed that the row and column labels are independent, meaning that $\mathbb{P}(Z^1 , Z^2) =\mathbb{P}(Z^1 )\mathbb{P}( Z^2) $ with $\mathbb{P}(Z^1) = \prod_{i,k} {\alpha_k}^{Z^1_{ik}} $ and $\mathbb{P}(Z^2) = \prod_{j,l} {\beta_l}^{Z^2_{jl}} $, where $\alpha_k = \mathbb{P}(Z_{i,k}=1), k = 1,\dots,Q_1$ and  $\beta_l = {\mathbb{P}(Z_{j,l}=1)}, {l = 1,\dots,Q_2}$
The density is thus equal to 

$$f_M (m, \theta_M) = \sum_{(Z^1,Z^2) \in (\mathcal{Z}^1,\mathcal{Z}^2 )}\prod_{i,k} {\alpha_k}^{Z^1_{ik}} \prod_{j,l} {\beta_l}^{Z^2_{jl}} \prod_{k,l} \prod_{i,j} \left( \pi_{kl}^{m_{i,j}}(1-\pi_{kl})^{1-m_{i,j}} \right)^{Z^1_{ik} Z^2_{jl}} $$

where $\mathcal{Z}^1 $ $\mathcal{Z}^2$  denoting the sets of possible labels $Z^1$ and $Z^2$.
The density of $m$ is parameterized by ${\theta_M = (\alpha,\beta,\pi)}$ with $\alpha = ( \alpha_1 , \dots,\alpha_{Q_1})  $, $\beta = ( \beta_1 , \dots,\beta_{Q_2})  $, and $\pi = \pi_{k,l};\quad k=1,\dots,Q_1;\quad l=1,\dots,Q_2$.
\\

Moreover, it is supposed that $N_{i,j} \sim \mathcal{P}(\lambda_i\mu_jG)$ with $\lambda_i \in ]0,1]$, $\mu_j \in ]0,1]$, with  $\max_{i=1,\dots,n_1}\lambda_i=1$ ,$\max_{j=1,\dots,n_2}\mu_j=1$ and $G>0$ .  The density of $N$ is then equal to 

$$f_N(n,\theta_N) = \prod_{i,j}  \frac{(\lambda_i\mu_jG)^{n_{i,j}}}{n_{i,j}!}e^{-\lambda_{i}\mu_{j}G}.$$

The density is parameterized by $\theta_N = (\lambda,\mu,G)$ with $\lambda = (\lambda_1,\dots,\lambda_{n_1})$ $\mu=(\mu_1,\dots,\mu_{n_2})$, and $G>0$.
\\

The density of $R_{i,j} = M_{i,j} N_{i,j}$ will depend both of the value of $Z^1,Z^2$ and of $\lambda_i, \mu_j$

$$
\mathbb{P}(R_{i,j} = r| Z^1_{ik}=1,Z^2_{jl}=1;\theta) = \left\{
    \begin{array}{ll}
        \ \pi_{kl} \frac{(\lambda_i\mu_jG)^{r}}{r!}e^{-\lambda_i\mu_jG}& \mbox{if } r>0 \\
        \\
        1 - \pi_{kl} ( 1 - e^{-\lambda_i\mu_jG})& \mbox{if }  r=0
    \end{array}
\right.
$$
 where $\theta = (\theta_M ,\theta_N) =  (\alpha,\beta,\pi,\lambda,\mu,G)$

\begin{align*}
f_R (r, \theta) = \sum_{(Z^1,Z^2) \in (\mathcal{Z}^1,\mathcal{Z}^2 )}\prod_{i,k} {\alpha_k}^{Z^1_{ik}} \prod_{j,l} {\beta_l}^{Z^2_{jl}} \prod_{k,l}\Bigg(& \prod_{\substack{i,j\\r_{i,j}>0}} \pi_{kl} \frac{(\lambda_i\mu_jG)^{r}}{r!}e^{-\lambda_i\mu_jG}\\
&\cdot \prod_{\substack{i,j\\r_{i,j}=0}} \left(1 - \pi_{kl} ( 1 - e^{-\lambda_i\mu_jG})\right)
\Bigg)\,. 
\end{align*}
% changer le nom en theorem 1 voir set counter...
\begin{theorem}
With $\pi$, the matrix of the Bernoulli coefficients, $\alpha$ and $\beta$, the row and column
mixing proportions of the mixture, $\lambda_i \in ]0,1]$, $\mu_j \in ]0,1]$, with  $\max_{i=1,\dots,n_1}\lambda_i=1$, $\max_{j=1,\dots,n_2}\mu_j=1$, $G>0$  and  assume that $n_1 \geq 2Q_2 -1$ and $n_2 \geq 2Q_1-1 $ 
\begin{itemize}
    \item (H1) for all $1 \leq k \leq Q_1,\alpha_k  > 0$ and the coordinates of vector $\tau = \pi \beta$ are distinct
     \item (H2) for all $1 \leq k \leq Q_2,\beta_k  > 0$ and the coordinates of vector $\sigma= \alpha' \pi $ are distinct (where
$\alpha'$ is the transpose of $\alpha$)
\end{itemize}

then the CoOP-LBM is identifiable.
\end{theorem}

\subsection{Proof}

The goal of the proof is to show that under assumption of Theorem 1 , there exist a unique parameter $\theta = (\alpha,\beta,\pi,\lambda,\mu,G)$, up to a permutation of row and column labels,
corresponding to $\mathbb{P}_\theta(r)$ the probability distribution function of matrix $r$ having at least $2Q_2 - 1$ rows and $2Q_1 - 1$ columns. The proof is in two parts : the first part deals with the identifiability of $(\lambda,\mu,G)$, the second part is adapted from the proof of Vincent Brault \textit{et al.} \cite{brault_estimation_2014}, deals with the identifiability of $(\alpha,\beta,\pi)$

\subsubsection{Identifiability of $(\lambda,\mu,G)$}
Let $\theta =(\alpha,\beta,\pi,\lambda,\mu,G)$ and $\theta ' = (\alpha',\beta',\pi',\lambda',\mu',G') $ two sets of parameters such as $\mathbb{P}_{\theta}=\mathbb{P}_{\theta'} $
Our first goal is to show that $(\lambda,\mu)$ = $(\lambda',\mu')$
As a reminder, $R_{i,j} = M_{i,j} N_{i,j}$. $M$ and $N$ are independant, we can deduce that $\mathbb{E}_\theta [R_{i,j}] = \mathbb{E}_\theta [M_{i,j}]\mathbb{E}_\theta [N_{i,j}]$.
Let  $\gamma = \mathbb{E}_\theta [M_{i,j}] = \sum_{k,l} \alpha_k \beta_l \pi_{kl}$ and $\gamma'  =  \mathbb{E}_{\theta'} [M_{i,j}]$ , we have

$$\mathbb{E}_\theta [R_{i,j}] = \lambda_i \mu_jG \gamma \quad = \quad \mathbb{E}_{\theta'} [R_{i,j}] = \lambda_i' \mu_j'G' \gamma'$$

\paragraph{Identifiability of $(\lambda,\mu)$}\mbox{}\\
Let 
 $\lambda = \lambda_1,\dots,\lambda_{n_1}$ with $\lambda_i \in ]0,1]$ and $\max\lambda_i = 1$. Without loss of generality we suppose that $\max\lambda_i = \lambda_1 = 1 $
 \\
 
 Let 
 $\lambda' = \lambda_1',\dots,\lambda_{n_1}'$ with $\lambda'_i \in ]0,1]$ and $\max\lambda'_i = 1$. Let $v \in \{1,\dots,n_1\}$ such that $\max\lambda'_i = \lambda'_v$ and now consider the quantities :
 \begin{align}
     &\mathbb{E}_\theta [R_{1,1}] = \lambda_1 \mu_1G \gamma \quad  \quad \mathbb{E}_{\theta'} [R_{1,1}] = \lambda_1' \mu_1'G' \gamma' \\
     &\mathbb{E}_\theta [R_{v,1}] = \lambda_v \mu_1G \gamma \quad  \quad \mathbb{E}_{\theta'} [R_{v,1}] = \lambda_v' \mu_1'G' \gamma'
 \end{align}

$\lambda_1 = 1$ and $ \lambda'_v=1$, moreover as $\mathbb{P}_{\theta}=\mathbb{P}_{\theta'} $ , we have that $\mathbb{E}_{\theta} [R_{1,1}] = \mathbb{E}_{\theta'} [R_{1,1}] $ and $\mathbb{E}_{\theta} [R_{v,1}] = \mathbb{E}_{\theta'} [R_{v,1}] $. From this we can see that $$\mathbb{E}_{\theta'} [R_{1,1}] = \lambda'_1 \mu'_1G' \gamma' = \lambda'_1 \mathbb{E}_{\theta'} [R_{v,1}] = \lambda'_1\mathbb{E}_{\theta}  [R_{v,1}] = \lambda_1'\lambda_v\mathbb{E}_{\theta}  [R_{1,1}]. $$

However, as $\mathbb{E}_{\theta} [R_{1,1}] = \mathbb{E}_{\theta'} [R_{1,1}] $, we can conclude that $\lambda_v \lambda_1'=1$ which implies that $\lambda'_1 = 1$ because both $\lambda_v\in ]0,1]$ and  $\lambda'_1\in ]0,1]$.
Knowing that $\lambda'_1 = 1 = \lambda_1$ we can conclude by seeing that $$\lambda'_i = \frac{\mathbb{E}_{\theta'} [R_{i,1}]}{\mathbb{E}_{\theta'} [R_{1,1}]} = \frac{\mathbb{E}_{\theta} [R_{i,1}]}{\mathbb{E}_{\theta} [R_{1,1}]} = \lambda_i$$ which proves that $\lambda = \lambda'$
Identifiability of $\mu$ is done in a similar way.

\paragraph{Identifiability of $G$}\mbox{}\\
Without loss of generality we suppose that $\max\lambda_i = \lambda_1 = 1 $ and $\max\mu_j = \mu_1 = 1 $, thanks to the precedent part we know that $\lambda = \lambda'$ and $\mu=\mu'$. Consider this quantities :

$$\frac{\mathbb{E}_{\theta} [R_{1,1}]}{\mathbb{P}_{\theta}(R_{1,1}>0)} \quad \text{and} \quad \frac{\mathbb{E}_{\theta'} [R_{1,1}]}{\mathbb{P}_{\theta'}(R_{1,1}>0)}  $$ which are equal because $\mathbb{P}_{\theta} = \mathbb{P}_{\theta'}$. We can calculate $$\frac{\mathbb{E}_{\theta} [R_{1,1}]}{\mathbb{P}_{\theta}(R_{1,1}>0)} = \frac{\lambda_1\mu_1G\gamma}{\gamma(1-e^{-\lambda_1\mu_1G})}=  \frac{G}{1-e^{-G}}\quad, $$ $x\mapsto \frac{x}{1-e^{-x}}$ is an injective function and therefore $G = G'$.

\subsubsection{Identifiability of $(\alpha,\beta,\pi)$}
In the precedent part we showed that given $\mathbb{P}_{\theta}$, $(\lambda,\mu,G)$ are defined in a unique manner, hence allowing us 
to determine $(\lambda,\mu,G)$. The following part is based from  \cite{brault_estimation_2014} and deals with the identifiability of $(\alpha,\beta,\pi)$.

\paragraph{Identifiability of $\alpha$}\mbox{}\\

Let $$\tau_k = (\pi \beta)_k = \sum_{l=1}^{Q_2} \pi_{k,l}\beta_l $$
 
 (H1) assures us that $(\tau_k)$ are distinct, therefore the $Q_1\times Q_1$ matrix $T$ defined by $T_{i,k} = (\tau_k)^{i-1}$ for $1\leq i\leq Q_1$ and $1\leq k\leq Q_1$
is Vandermonde, and hence invertible. Consider now $u_p$, the
probability to have a number greater than 0 on the $p$ first cells of the first row of $r$.
\begin{align}
    u_p & = \mathbb{P}(r_{1,1}>0,\dots,r_{1,p}>0)\\
    &= \sum_{k, l_1,\dots,l_p}\mathbb{P}(Z^1_{1k}=1)\prod_{j=1}^p\left[\mathbb{P}(r_{1,j}>0|Z^1_{1k}=1, Z^2_{jl_j}=1)\mathbb{P}(Z^2_{jl_j}=1) \right]\\
    &= \sum_{k=1}^{Q_1}\mathbb{P}(Z^1_{1k}=1)\prod_{j=1}^p\left[\sum_{l_j=1}^{Q_2}
    \mathbb{P}(r_{1,j}>0|Z^1_{1k}=1, Z^2_{jl_j}=1)\mathbb{P}(Z^2_{jl_j}=1) \right]\\
    &=\sum_{k=1}^{Q_1}\mathbb{P}(Z^1_{1k}=1)\prod_{j=1}^p\left[\sum_{l_j=1}^{Q_2}\beta_{j_l} 
    \pi_{k,l_j}(1-e^{-\lambda_1\mu_jG})\right]\\
    &=\sum_{k=1}^{Q_1}\mathbb{P}(Z^1_{1k}=1)\prod_{j=1}^p\left[\sum_{l_j=1}^{Q_2}\beta_{j_l} 
    \pi_{k,l_j}\right] \prod_{j=1}^p \left[(1-e^{-\lambda_1\mu_jG})\right]\\
    &=\prod_{j=1}^p \left[(1-e^{-\lambda_1\mu_jG})\right]\sum_{k=1}^{Q_1}\mathbb{P}(Z^1_{1k}=1)\left[\sum_{l_j=1}^{Q_2}\beta_{j_l} 
    \pi_{k,l_j}\right]^p \\
    &=\prod_{j=1}^p (1-e^{-\lambda_1\mu_jG})\sum_{k=1}^{Q_1}\alpha_k(\tau_k)^p
\end{align}
with a given $\mathbb{P}(r)$, $u_1,\dots,u_{2Q_2-1}$ are known and $K_p = \prod_{j=1}^p (1-e^{-\lambda_1\mu_jG}) >0 $ are also known thanks the indentifiability of $(\lambda,\mu,G)$. We denote $u_0 =1$, $K_0 = 1$  and let now $U$ be the $(Q_1+1) \times Q_1$ matrix defined by
$$U_{i,j} = \frac{u_{i+j-2}}{K_{i+j-2}} = \sum_{k=1}^{Q_1}(\tau_k)^{i-1}\alpha_k(\tau_k)^{j-1} $$ and let $U_i$ be
 the square matrix obtained by removing the row $i$ from $U$. We can see if we write $A = diag(\alpha)$ that
 
 $$U_{Q_1+1} = TAT' .$$
 
 $T$ is unknown at this stage, but can be found by noticing that the coefficient $\tau$ are the root of the following polynomial \citep{celisse_consistency_2012} : 
 
 $$G(x)=\sum_{k=0}^{Q_1}(-1)^{k+Q_1}D_{k+1} x^{k}$$
 where $D_{k+1} = detU_{k+1}$ and $D_{Q_1+1}\neq 0$ because $U_{Q_1+1}$ is the product of invertible matrices. Consequently it is possible to determine $\tau$ and $T$. As $T$ is invertible, we can conclude that $A = T^{-1}U_{Q_1+1}T'^{-1}$ is defined in an unique manner.

\paragraph{Identifiability of $\beta$}\mbox{}\\

The identifiability of $\beta$ is done the same way as the identifiability of $\alpha$, by considering hypothesis (H2) and

$$\sigma_l = \sum_{k=1}^{Q_1} \pi_{k,l}\alpha_k,\quad  S_{j,l} = (\sigma_l)^{j-1}
$$

$$v_q  = \mathbb{P}(r_{1,1}>0,\dots,r_{q,1}>0), \quad H_q = \prod_{i=1}^q (1-e^{-\lambda_i\mu_1G}) >0  $$ 
$$V_{i,j} = \frac{v_{i+j-2}}{H_{i+j-2}},\quad B=diag(\beta) $$ and showing that $$B = S^{-1}V_{Q_2+1}S'^{-1} $$

\subsubsection{Identifiability of $\pi$}

Consider now $w_{pq}$, the
probability to have a number greater than 0 on the $p$ first cells of the first row of $r$ and $q$ first cells of the first column of $r$ with $1\leq p \leq Q_1$ and $1\leq q \leq Q_2$

\begin{align}
    w_{pq} &= \mathbb{P}(r_{1,1}>0,\dots,r_{1,p}>0,r_{2,1}>0,\dots r_{q,1}>0)\notag\\
    &= \sum_{k=1}^{Q_1}\sum_{l=1}^{Q_2}\alpha_k\beta_l \pi_{kl}(1-e^{-\lambda_1\mu_1G})
    \prod_{j=2}^p\left[\sum_{l_j=1}^{Q_2}\beta_{j_l} 
    \pi_{k,l_j}(1-e^{-\lambda_1\mu_jG})\right]
    \prod_{i=2}^q\left[\sum_{k_i=1}^{Q_1}\alpha_{i_k} 
    \pi_{k_i,l}(1-e^{-\lambda_i\mu_1G})\right]\notag\\
    &=  \frac{K_p}{K_1} \frac{H_q}{H_1}(1-e^{-\lambda_1\mu_1G}) \sum_{k=1}^{Q_1}\sum_{l=1}^{Q_2}\alpha_k (\tau_k)^{p-1} \pi_{kl}
    \beta_l (\sigma_l)^{q-1}\notag\\
    &= \psi_{pq}   \sum_{k=1}^{Q_1}\sum_{l=1}^{Q_2}\alpha_k (\tau_k)^{p-1} \pi_{kl}
    \beta_l (\sigma_l)^{q-1}\notag
\end{align}
where $\psi_{pq} = \frac{K_p}{K_1} \frac{H_q}{H_1}(1-e^{-\lambda_1\mu_1G}) >0$ is known thanks the indentifiability of $(\lambda,\mu,G)$\\

Let $W_{p,q} = \frac{w_{pq}}{\psi_{pq}}$ we can see that 
$$W = RA\pi BS' $$ thus 

$$\pi = A^{-1}R^{-1}WS'^{-1}B^{-1} $$ is defined in a unique manner.